\documentclass{sig-alternate}
\usepackage{graphicx}
\usepackage{balance}
 \usepackage{url}
 \usepackage{xcolor}
\usepackage{verbatim}
\usepackage{hyperref}
\usepackage{epsfig}
\usepackage{latexsym}
\usepackage{times}
\usepackage[T1]{fontenc}
\usepackage{textcomp}
\usepackage{enumitem}
\usepackage[nocompress]{cite}
\usepackage{relsize}
\usepackage{framed}
\usepackage[scriptsize, it, IT]{subfigure}
\usepackage{mathtools}
\usepackage[font={scriptsize,it}]{caption}

\usepackage{algpseudocode}
\usepackage{algorithm}

\begin{document}

\newtheorem{theorem}{Theorem}

\newenvironment{denselist}{
    \begin{list}{\tiny{$\bullet$}}%
    {\setlength{\itemsep}{0ex} \setlength{\topsep}{0ex}
    \setlength{\parsep}{0pt} \setlength{\itemindent}{0pt}
    \setlength{\leftmargin}{1.5em}
    \setlength{\partopsep}{0pt}}}%
    {\end{list}}

\newcommand{\squishlist}{
   \begin{list}{$\bullet$}
    { \setlength{\itemsep}{0pt}
      \setlength{\parsep}{2pt}
      \setlength{\topsep}{2pt}
      \setlength{\partopsep}{0pt}
    }
}
\newcommand{\squishend}{\end{list}}
\newcommand{\eat}[1]{}
\newcommand{\papertext}[1]{}
\newcommand{\techreporttext}[1]{#1}
\newcommand{\agp}[1]{\noindent{\textcolor{red}{Aditya: #1}}}
\newcommand{\revision}[1]{\textcolor{red}{#1}}
\newcommand{\yihan}[1]{\noindent{\textcolor{red}{Yihan: #1}}}
\newcommand{\stitle}[1]{\vspace{0.5em}\noindent\textbf{#1}}

\renewcommand{\baselinestretch}{0.97}

%

\title{Finish Them!: Pricing Algorithms for Human Computation}
%
%
%
%
%

\numberofauthors{2} 
%
\author{
%
%
\alignauthor
Yihan Gao\\
       \affaddr{University of Illinois (UIUC)}\\
       \affaddr{Urbana, Illinois}\\
       \email{ygao34@illinois.edu}
\alignauthor
Aditya Parameswaran\\
       \affaddr{University of Illinois (UIUC)}\\
        \affaddr{Urbana, Illinois}\\
       \email{adityagp@illinois.edu}
}

\maketitle
\begin{abstract}
Given a batch of human computation tasks, a commonly ignored aspect is how the price (i.e., the reward paid to human workers) of these tasks must be set or varied in order to meet latency or cost constraints. Often, the price is set up-front and not modified, leading to either a much higher monetary cost than needed (if the price is set too high), or to a much larger latency than expected (if the price is set too low).
Leveraging a pricing model from prior work, we develop algorithms to optimally set and then vary price over time in order to meet a (a) user-specified deadline while minimizing total monetary cost  (b) user-specified monetary budget constraint while minimizing total elapsed time. We leverage techniques from decision theory (specifically, Markov Decision Processes) for both these problems, and demonstrate that our techniques lead to upto 30\% reduction in cost over schemes proposed in prior work. Furthermore, we develop techniques to speed-up the computation, enabling users to leverage the price setting algorithms on-the-fly. 
\end{abstract}





\section{Introduction}
Crowdsourcing is often used to process and reason about unstructured data such as images, videos, and text. The data thus generated is typically used as training data for machine learning algorithms in applications such as content moderation (i.e., determining if images are suitable to be viewed by a general audience), spam detection, search relevance estimation, information extraction, and entity resolution. In fact, all of the following companies employ crowdsourcing frequently at a large scale to repeatedly process unstructured data: Google~\cite{samasource-crowd}, Ebay~\cite{ebay-crowd}, Microsoft~\cite{crowd-art}, LinkedIn~\cite{goog-fb-linkedin}, Facebook~\cite{goog-fb-linkedin}, Yahoo!~\cite{psox}, Twitter~\cite{raven}, Cisco~\cite{samasource-crowd}, and Yelp~\cite{yelp-crowd}.

Even though crowdsourcing is often used in industry and academia, and has been the subject of many academic papers studying tradeoffs between cost, latency and accuracy~\cite{crowddb,deco,hqueryCIDR,crowdscreen,karger}, there is little to no work on {\em task pricing} and its impact on overall cost and latency: that is, how the price of tasks (i.e., the monetary reward paid to workers on completion) must be set or varied in order to meet cost or latency constraints. Often, the price is set up-front and not modified, leading to either a much higher monetary cost than needed (if the price is set too high), or to a much larger latency than expected (if the price is set too low). As a result, anecdotally, pricing is seen as somewhat of a ``dark art''.

In this paper, we wish to address the following question: {\em Given that we have $n$ fixed tasks, how should we vary their price or reward over time so that they get completed by a certain deadline at the least cost possible?} 
Intuitively, it seems that we may want to start with a low price initially, and then increase it gradually as it gets closer to the deadline. However, there has been no work demonstrating that such strategies will indeed yield good results in practice. Furthermore, there are a number of additional complications, even given this very simple scheme: 
\begin{denselist}

\item What should we price tasks initially?

\item How can we adapt our price setting to the rate at which tasks are picked up? What if tasks get picked up very quickly at the initial price; should we lower the price, should we keep it same, or should we increase it? What if the opposite happens --- that is, tasks get picked up very slowly at the initial price?

\item At what time points should we increase the price? Increasing it too frequently may lead to computationally more expensive decision making (as we will see subsequently), but increasing it too infrequently may result in much higher costs.

\item At what granularities do we increase the price, and how much does this affect overall cost? 

\item Should we price all the tasks the same, or should we price tasks differently?

\item What if we had a fixed budget, and instead wanted to reduce total latency. Would similar techniques apply then? Would varying price help at all?

\item How do we ensure that our pricing schemes can be computed within a reasonable time, and how can we speed them up?

\item How are our algorithms impacted by inaccuracies in estimates of the marketplace dynamics?

\end{denselist}
In prior work,  Faridani et al.~\cite{faridani11} develop a model for latency in crowdsourcing applications based on Non-Homogeneous Poisson Processes. They then use this model to describe a simple scheme based on binary search for pricing tasks to complete by a deadline. However, their scheme is {\em not optimal}, that is, it wastes far too much monetary cost. In this paper, we leverage their model and instead focus on the optimization problem of minimizing cost while meeting the deadline with high probability. Overall, our techniques yield rich dividends --- we get up to a 30\% reduction in cost as compared to their scheme on realistic crowdsourcing workloads. This represents a significant reduction in cost especially for users who run large crowdsourcing workloads with strict deadlines. 

In this paper, we develop algorithms for two optimization problems, given a set of tasks: one, minimizing cost while meeting time requirements, and second, minimizing latency while meeting monetary budget requirements. For the first, we develop an algorithm based on decision theory that gives us near-optimal results. For the second, we develop a solution that uses linear programming, that can be shown to be optimal under some assumptions. A crucial concern for us is that the computation is as little as possible, and we propose various speed-up techniques for this purpose. 

The contributions of this paper are as follows:
\begin{denselist}

\item We describe the two problems that we study in this paper formally in Section~\ref{sec_model}.

\item We develop optimized pricing algorithms that meet a fixed time deadline in Section~\ref{sec_fix_deadline}. Since these algorithms could be computationally expensive, we describe techniques to reduce the complexity of these algorithms.

\item We develop optimized pricing algorithms that meet a fixed monetary cost budget in Section~\ref{sec_fixed_budget}.

\item We demonstrate that our pricing algorithms achieve a reduction in cost of up to 30\% over prior work on simulations with real data from a crowdsourcing marketplace, as well as live experiments on the same marketplace in Section~\ref{sec_experiments}. Furthermore, we demonstrate that the algorithms are remarkably robust to errors in the estimates of parameters of the tasks and the marketplace.

\end{denselist}
We cover related work in Section~\ref{sec:related} and conclude in Section~\ref{sec_conclusion}.

\section{Preliminaries}\label{sec_model}

In this section, we describe the basic model that we will leverage to design optimized pricing algorithms.

We operate on a crowdsourcing marketplace, such as Mechanical Turk~\cite{mturk}. In any crowdsourcing marketplace, users (or requesters) post {\em tasks}, often many at a time, and set a monetary {\em price or reward} for them. At any point, there are many tasks on offer in the marketplace. Human {\em workers} arrive at the marketplace at any time, and can leave at any time. When on the marketplace, workers can choose to work on any of the available tasks. They are allowed to work on a single task at a time. Once they complete a task, they will receive the reward or price assigned for the task by the requester.

In a marketplace, the reward of each task is positively correlated with the completion rate: the higher the reward, the shorter the completion time. However, in order to determine the best trade-off between cost and completion time, this relationship must be precisely quantified. For example, we must be able to answer questions like: if we adjust the reward per task from \$$0.25$ to \$$0.3$, how much do we gain in terms of task completion rate? To answer these questions, we need a formal model for reasoning about the crowdsourcing marketplace. 

In previous work, Faridani et al.~\cite{faridani11} studied the problem of modeling crowdsourcing marketplace dynamics; the dynamics is modeled using two independent processes: A \textit{Non-Homogeneous Poisson Process} is used to model the worker arrivals in the market, and a \textit{Discrete Choice Model} is used to model how workers choose between tasks in the marketplace. We adopt the same model in this paper, and focus instead on the optimal pricing problem. To enable this paper to be self-contained, we describe the worker arrival model in Section~\ref{sec_worker_arrival}, and the task choice model in Section~\ref{sec_prob_est}. These mathematical models will be used to define the pricing problem formally in Section~\ref{sec_prob_state}. 

\subsection{Worker Arrival Model}\label{sec_worker_arrival}

Faridani et al.~\cite{faridani11} show that the arrival of workers in a crowdsourcing marketplace follows a \textit{Non-Homogeneous Poisson Process(NHPP)}. Note that the standard Poisson process is commonly used to characterize the counting process of stochastically occurring events. The  Poisson process has a fixed rate $\lambda$. NHPP is a generalization of the Poisson process, with a rate parameter $\lambda(t)$, a function of time~\cite{ross1996stochastic}. In a NHPP, the number of events that occur during any period of time $[S, T]$ follows a Poisson distribution:
\begin{equation}\label{eqn_poisson_proc}
\mathbf{N}[S,T] \sim \mathbf{Pois}(\cdot|\lambda = \int_{t = S}^T \lambda(t))
\end{equation}
where $\mathbf{Pois}(\cdot|\lambda)$ refers to a Poisson distribution with mean $\lambda$.

Estimating the arrival-rate function $\lambda(t)$ of a NHPP is more difficult than that for a Homogeneous Poisson Process because of the infinite dimensionality of the arrival-rate parameter $\lambda(t)$. Therefore, a common approach is to assume a parametric form for $\lambda(t)$. For instance, Massey et al.~\cite{massey1996estimating} used a piece-wise linear function to approximate the traffic of telecommunication systems.

Figure~\ref{fig_mturk-track_hit_completion} depicts the number of tasks completed every $6$ hours for a time range of 4 weeks in Mechanical Turk. The figure depicts that the variation of worker arrivals follows a process that approximately repeats every week. In this paper, we assume that the arrival-rate function $\lambda(t)$ is periodic, and the variations in the number of worker arrivals are all due to the randomness of the Poisson process. Given historical data, the arrival-rate function $\lambda(t)$ can be estimated and used to predict arrival rates in the future. Faridani et al.~\cite{faridani11} provide techniques for learning the $\lambda(t)$ function, and demonstrate the accuracy of these techniques. In this paper, we leverage these techniques, and assume that $\lambda(t)$ is known. As we will see in our experimental results in Section~\ref{sec_experiments}, our pricing strategies are not very sensitive to mistakes in the estimation of $\lambda(t)$.

\begin{figure}[h]
\vspace{-8pt}
\centering
\includegraphics[width = 2.3in]{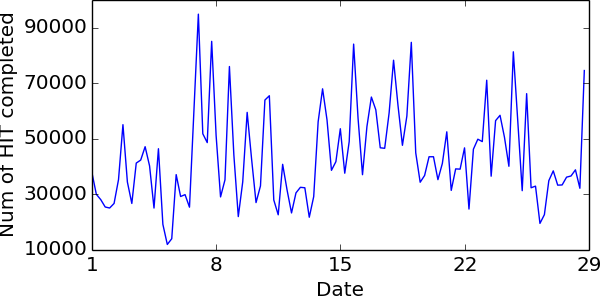}
\vspace{-10pt}
\caption{The number and total value of tasks completed each day from January 1st, 2014 to January 28th, 2014. Retrieved from http://www.mturk-tracker.com}
\label{fig_mturk-track_hit_completion}
\vspace{-10pt}
\end{figure}

Note that the NHPP models the arrival of workers to the entire marketplace, and does not capture whether those workers decide to work on our specific task. An independent Bernoulli process can be used to model whether each worker (who arrives at the marketplace) will decide to work on our task. In other words, we assume that each arrived worker has an independent probability $p$ of picking our task. Therefore if in any period of time the number of workers arrived at marketplace is $X$, then, assuming there are adequate tasks on offer, the number of workers from those who choose to work on our tasks will follow a Binomial distribution $\mathbf{Bin}(X, p)$. The value of $p$, the task acceptance probability, is not directly observable. We describe how it is related to the price or reward for the task, and how it can be estimated in Section~\ref{sec_prob_est}.

The task completion process is then a composition of NHPP and an independent Bernoulli process. In Statistics literature, such a process is called a \textit{Thinned Non-Homogeneous Poisson Process}~\cite{ross1996stochastic}. A Thinned NHPP is also a NHPP with a modified arrival-rate function $\lambda'(t) = \lambda(t) p$. 

\subsection{Task Acceptance Probability Estimation}\label{sec_prob_est}

Faridani et al.~\cite{faridani11} used a \textit{Discrete Choice Model} to characterize how workers select tasks from the marketplace. In economics, \textit{Discrete Choice Models} are used to estimate the probability of consumers choosing a specific product among a range of alternatives~\cite{mcfadden1973conditional}. Discrete Choice Models can be explained by utility theory: each worker chooses the task in the marketplace to maximize the utility (or net benefit) obtained. Workers may have different perceptions of his/her utility: it could depends on various factors such as hourly wage, number of tasks, task type, easiness of the tasks, or the knowledge gain during the process of finishing a task. Utility can not be directly observed, the only aspect that can be observed is the worker's behavior in the marketplace.

Under this model, the task acceptance probability parameter $p$ is simply the probability that the utility of our task exceeds the utility of every other task in the marketplace. Let $U_i$ be the utility of task $i$ in the marketplace based on some worker's perception, and without loss of generality we assume the utility of our task is $U_1$. Then
$p = \mathbf{Pr}(U_1 > \max_{i \not= 1} U_i)$.
In the \textit{Conditional Logit Model}~\cite{mcfadden1973conditional}~\cite{faridani11}, the utility $U_i$ of $i$th task has the following expression: $U_i = \mathbf{\beta}^T z_i + \epsilon_i$,
where $z_i$ are all observable attributes that may affect the utility of the task, and $\epsilon_i$ accounts for all unobserved factors that may affect the utility. In the model, the utility $U_i$ is assumed to be linearly correlated with all observed attributes with the shared coefficient vector $\beta$. The parameters $\epsilon_i$ are assumed to be independent with each other and follow the Gumbel distribution. Based on these assumptions, it can be derived that the probability of choosing each task follows a \textit{Multinomial Logit Distribution}:
$$p = \mathbf{Pr}(U_1 > \max_{i \not= 1} U_i) = \frac{\exp(\beta^T z_1)}{\sum_i \exp(\beta^T z_i)}$$

Now if we are able to change our task reward $c$, then the attribute vector of our task $z_1$ and the task acceptance probability $p$ will also change accordingly:
\begin{equation}\label{eqn_accept_prob}
p(c) = \frac{\exp(\beta^T z_1(c))}{\exp(\beta^T z_1(c)) + \sum_{i \not= 1} \exp(\beta^T z_i)}
\end{equation}


Equation~(\ref{eqn_accept_prob}) captures how task acceptance probability is related to task price or reward. Faridani et al.~\cite{faridani11} suggest using this equation directly in order to calculate the task acceptance probability, with parameters $\beta$ estimated from historical marketplace data using logistic regression. Another approach is to assume a parametric form of task acceptance probability function, and estimate parameters during a separate training phase. If we assume that the utility of our task is a linear function of task reward $c$, and that the sum of exponentials of the utilities of the other tasks is a fixed constant, then Equation~(\ref{eqn_accept_prob}) can be rewritten as:
\begin{equation}\label{eqn_simplified}
p(c) = \frac{\exp\{\frac{c}{s} - b\}}{\exp\{\frac{c}{s} - b\} + M}
\end{equation}
Hence if we have some training data (e.g., estimated value of $p(c)$ for different task reward $c$), then parameters $s, b, M$ can be estimated by statistical regression methods.

\eat{The parameters $\beta$ can be estimated from historical marketplace data using logistic regression~\cite{faridani11}.}
However, note that the inference of the mapping function $p(c)$ is not the focus of our paper. Here, we will assume that the expression of $p(c)$ is already known\eat{, using techniques from Faridani~\cite{faridani11}}. We will then use this expression to determine the optimal reward for each task in various scenarios.

\subsection{Problem Statement}\label{sec_prob_state}
Our goal is to design pricing algorithms for batch of $N$ identical crowdsourcing tasks. The user may specify either a monetary budget restriction (that is, the algorithm must ensure that all tasks are completed within a certain expected cost), or a time deadline (that is, the algorithm must ensure that all tasks are completed within a certain time). The unconstrained variable (monetary cost or overall time) is minimized. 

Following our discussion in the previous section, we model worker arrivals to the marketplace as a \textit{Non-Homogeneous Poisson Process} with a known arrival-rate parameter $\lambda(t)$. Each worker will pick up our task and complete it with probability $p(c)$, where the value of $p$ depends on the reward $c$ (typically in cents or dollars) for each task in our batch of tasks. The form of the mapping function $p(c)$ from task reward $c$ to task acceptance probability $p$ is assumed to be known: thus, we expect our techniques to be leveraged when the user ends up repeating similar tasks many times over a long period so that such history is available. This is not a drastic assumption to make: many companies, including Google, Ebay, Yahoo!, and Microsoft, repeatedly use human workers for tasks such as content moderation, categorization, spam detection, and search relevance.


At any time, we can monitor the number of remaining uncompleted tasks $n$. The task reward $c$  can be changed at any time, and the task acceptance probability $p$ will change accordingly. Note that some marketplaces may impose a minimum time only after which the task reward may be changed, and our algorithms adapt to that scenario as well. Overall, at any time $t$, the completion of tasks follows a NHPP with rate $\lambda(t) p(c)$, and for each completed task, $c$ units of monetary compensation are paid based on the task reward at that time.

Then, the problem is to determine and dynamically vary the rewards for each as yet unsolved task, such that the total monetary cost expended and the total time used for completing $N$ tasks are minimized. We focus on two scenarios:
\begin{denselist}
\item \textbf{Fixed Deadline Pricing} (Section~\ref{sec_fix_deadline}): In this scenario, the total time used to complete all $N$ tasks must be less than a deadline $T$. The goal is then to minimize the expected total expenditure. 
\item \textbf{Fixed Budget Pricing} (Section~\ref{sec_fixed_budget}): In this scenario, the total monetary budget $B$ for tasks is fixed upfront. The goal is then to minimize the expected total time to complete all tasks. 
\end{denselist}
\papertext{In our extended technical report~\cite{treport}, we describe a number of straightforward generalizations, including optimizing combinations of deadline and budget, capturing multiple task types, and incorporating accuracy and difficulty. In short, optimizing for combinations of deadline and budget is actually a simplification of the techniques described in the paper; capturing multiple tasks types is a straightforward generalization of the results of \cite{crowdscreen}; and incorporating accuracy and difficulty involves using results from \cite{crowdscreen} to first optimize for accuracy, then using that information while developing our techniques---we develop one exact, but intractable technique and two approximations for this case.}
\techreporttext{In Section~\ref{sec:discussion}, we describe a number of straightforward generalizations, including optimizing combinations of deadline and budget, capturing multiple task types, and incorporating accuracy and difficulty.}


\section{Fixed Deadline Pricing Strategy}\label{sec_fix_deadline}

It is common for task requesters in a crowdsourcing marketplace to require their tasks to be completed before a certain deadline. Under this scenario, the reward for each task in the batch of tasks should be as low as possible while making sure that all tasks can be completed before deadline.

In Faridani's work~\cite{faridani11}, a binary search process is used to find the smallest fixed task reward such that the total expected completion time is before the deadline. However, as implied by the NHPP worker-arrival model and also demonstrated in Figure~\ref{fig_mturk-track_hit_completion}, the task completion process is highly non-deterministic. Therefore, a dynamic pricing strategy should perform much better in terms of overall cost in this scenario: If the rate at which tasks are picked up by workers is faster than expected, we could decrease the reward for the remaining tasks to save money; on the other hand, if the tasks are picked up slower than expected, we could increase the reward to attract more workers to our tasks.

In this section, we design a pricing algorithm to determine how to set the reward for each task at each time point to minimize the expected total monetary cost, while meeting time constraints. We begin by modeling our decision process as a Markov process and use the model to present our basic pricing algorithm in Section~\ref{sec_mdp}. Since these algorithms may be expensive to compute, we present techniques that can help speed up the computation in Section~\ref{sec_speedup}. 
\techreporttext{Lastly, we consider different objectives in Section~\ref{sec_objectives}.}

\subsection{Markov Decision Process-based Solution}\label{sec_mdp}
\vspace{-3pt}
\stitle{Discretization:} Although in principle we may be able to change the task reward $c$ at any time, utilizing this freedom while designing pricing strategies would result in an intractable number of time points at which decisions need to be made. Instead, we discretize the total time before the deadline (i.e., the time between when the tasks were submitted to the marketplace and the deadline) into a number of equal-sized intervals. As we will see later on, beyond a point, discretization does not help, and therefore restricting our pricing algorithms to make decisions only at discrete time intervals does not affect the overall monetary cost, while significantly reducing the computation involved.

We partition all available time (from $t = 0$, i.e., start time, to $t = T$, i.e., the deadline) $[0, T]$ into $N_T$ small intervals: $[0, T / N_T), $ $[T / N_T, 2T / N_T), $ $\ldots, $ $[T - T / N_T, T)$ and further enforce that the reward $c$ for tasks may only be changed at the start of an interval.

\stitle{State Space:} After discretization, we can represent the state of processing of the batch of tasks at any time interval using a finite Markov chain. The states in this Markov chain are represented by a pair $(n, t)$, where $n$ is the remaining unsolved tasks and $t$ is the index of current time interval. The initial state is $(N, 0)$ and all states in the form of $(n, N_T)$ are final states (Recall that $N_T$ is the total number of time intervals).

An illustration of the state diagram is shown in Figure~\ref{fig_state_diagram}. The states are represented on a grid, where the number of unsolved tasks increases along the $y$-axis, and the number of time intervals elapsed increases along the $x$-axis. Our goal is then to set the prices $c_{n, t}$ upfront for all $n, t$, such that we have as few unsolved tasks as possible when $t = N_T$.

\stitle{Transitions:} Based on Equation~(\ref{eqn_poisson_proc}), $X_t$, the number of tasks completed during the $t$th time interval follows a Poisson distribution:
$X_i \sim \mathbf{Pois}(\cdot|\lambda = \lambda_t p(c_t))$
where $c_t$ is the task reward in $i$th time interval, and $\lambda_t$ is the total expected number of workers who arrived at marketplace during the $t$th time interval:
\begin{equation}\label{eqn_lambda_time_segment}
\lambda_t = \int_{s = (t - 1)T/N_T}^{tT/N_T} \lambda(s) ds
\end{equation}
At state $(n, t)$, say the task reward is set to be $c_{n, t}$; then, the transition probability between states is:

{\scriptsize
\begin{align}\label{eq:probtrans}
\mathbf{Pr}\{(n, t) \rightarrow (n - s, t + 1)|c_{n,t}\} & = \mathbf{Pois}(s|\lambda = \lambda_t p(c_{n,t}))\\ & = e^{-\lambda_t p(c_{n,t})} \frac{(\lambda_t p(c_{n,t}))^s}{s!}
\end{align}
}
\noindent where $\lambda_i$ is defined in Equation~(\ref{eqn_lambda_time_segment}) and $p(c_{n,t})$ is the task acceptance probability for the task reward $c_{n,t}$. The transition probability is slightly different when we are close to completion:
$$\mathbf{Pr}\{(n, t) \rightarrow (0, t + 1)|c_{n,t}\} = \mathbf{Pr}(\mathbf{Pois}(\cdot|\lambda = \lambda_t p(c_{n,t})) \geq n)$$

\begin{figure}[h]
\centering
\includegraphics[width = 2in]{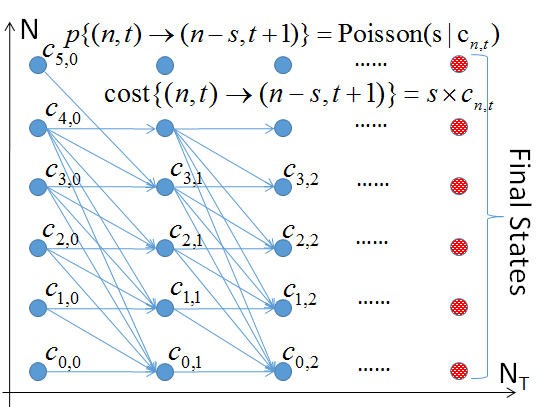}
\vspace{-5pt}
\caption{State diagram of Markov Decision Process. Some possible transitions are omitted in the figure for clarity.}
\label{fig_state_diagram}
\vspace{-10pt}
\end{figure}

\stitle{Costs:} In our problem, the transition cost between states is the total rewards paid for tasks completed in each time interval:
\begin{equation}
\label{eq:costtrans}
\mathbf{cost}\{(n, t) \rightarrow (n - s, t + 1)|c_{n,t}\} = s c_{n,t}
\end{equation}

For the final states $(n, N_T)$, we assign a fixed penalty for each of the remaining unsolved tasks:
$\mathbf{cost}\{(n, N_T)\} = n \times \mathbf{Penalty}$,
where the value of the parameter $\mathbf{Penalty}$ could be based on actual expenses needed to complete them post deadline (possibly by the task requester themselves), or simply be set large enough to ensure that with high probability no task will remain uncompleted.

\papertext{
In the above formulation, we have combined the number of unsolved tasks by the deadline and the monetary cost incurred into one objective. In our extended technical report~\cite{treport}, we consider a range of other constraints and objectives, including:
\begin{denselist}
\item Optimizing monetary cost while enforcing a constraint on the probability that there are additional tasks left after the deadline
\item Optimizing monetary cost while enforcing a constraint on the expected number of tasks left after the deadline
\end{denselist}
The techniques for these formulations are straightforward modifications of the techniques for the current formulation.
}

\stitle{Markov Decision Processes:} The problem of determining optimal task reward $c_{n,t}$ in $t$th time interval for state $(n, t)$ can be viewed as a \textit{Markov Decision Process (MDP)}. MDPs are commonly used to model optimization and decision making problems in a discrete time stochastic environment. The goal is of MDP optimization to determine the policy for every state to minimize the expected overall cost (in our problem it corresponds to determining the optimal task reward $c_{n,t}$ for each state).

\stitle{Dynamic Programming:} The above MDP optimization problem can be solved by \textit{Dynamic Programming (DP)}. Let $\mathbf{Opt}(n, t)$ denote the minimum expected total cost for all remaining $n$ tasks for the state $(n, t)$, and $\mathbf{Price}(n, t)$ denote the corresponding optimal reward for each task. Then $\mathbf{Opt}(n, t)$ and $\mathbf{Price}(n, t)$ satisfy the following equations.
{\scriptsize
\begin{align*}
\mathbf{Opt}(n, t) = \min_c \sum_{s=0}^n & [\mathbf{Opt}(n - s, t + 1) + s c] \\
& \times \mathbf{Pr}\{(n, t) \rightarrow (n - s, t + 1)|c\}\\
\mathbf{Price}(n, t) = \arg \min_c \sum_{s=0}^n & [\mathbf{Opt}(n - s, t + 1) + s c] \\
& \times \mathbf{Pr}\{(n, t) \rightarrow (n - s, t + 1)|c\}
\end{align*}
}
The values of $\mathbf{Opt}(n, t)$ and $\mathbf{Price}(n, t)$ can be sequentially determined\eat{ since the underlying Markov Chain is acyclic}.
That is, we start at $(\cdot, N_T)$, and work our way backwards using the equations above. Once we have computed the optimal $\mathbf{Opt}$ and $\mathbf{Price}$ for all $(\cdot, t+1)$, we can use the equations above to compute it for all $(\cdot, t)$ --- the optimal $c_{n,t}$ can be found by considering all possible price values since it needs to be an integral multiples of a minimal unit of price (In Amazon Mechanical Turk it is $1$ cent). 
\papertext{The pseudocode for the DP algorithm can be found in the extended technical report~\cite{treport}.}
\techreporttext{Algorithm~\ref{alg_simple_dp} gives the pseudocode of this DP algorithm.}

\techreporttext{
\begin{algorithm}
\scriptsize
\caption{Simple Dynamic Programming}
\label{alg_simple_dp}
\begin{algorithmic}
\Function{FindOptimalPriceForState}{$n$, $t$, $L$, $U$}
	\State $Opt(n, t) \leftarrow \infty$
	\For {$c = L$ to $U$}
		\State $Cost \leftarrow 0$, $Pr \leftarrow 0$
		\State $AcceptRate \leftarrow p(c)$
		\For {$i = 0$ to $n$}
			\State $p \leftarrow \mathbf{Pois}(i|\lambda(t) \times AcceptRate)$
			\State $Cost \leftarrow Cost + p \times (i c + Opt(n - i, t + 1))$
			\State $Pr \leftarrow Pr + p$
		\EndFor
		\State $Cost \leftarrow Cost + (1 - Pr) \times nc$
		\If {$Cost < Opt(n, t)$}
			\State $Opt(n, t) \leftarrow Cost$
			\State $Price(n, t) \leftarrow c$
		\EndIf
	\EndFor
\EndFunction

\Function{SimpleDP}{}
	\For{$i = 0$ to $N$}
		\State $Opt(i, N_T) \leftarrow i \times \mathbf{Penalty}$
	\EndFor

	\For{$t = N_T - 1$ to $0$}
		\For{$i = 0$ to $N$}
			\State FindOptimalPriceForState($i$, $t$, $0$, $C$)
		\EndFor
	\EndFor
\EndFunction
\end{algorithmic}
\end{algorithm}
}

\subsection{Speed-up Techniques}\label{sec_speedup}

The DP algorithm has a time complexity of $O(N^2 N_TC)$, where $C$ is the number of price choices we want to consider, which is intractable when $N$ is large or when $N_T$ or $C$ are fine-grained. Here we discuss some techniques to speed up the algorithm.

\stitle{Poisson Distribution Truncation:}
Notice that while making pricing decisions, the DP algorithm enumerates all possible number of tasks $s$ that can be picked up by workers during each time interval. However, for large $s$, the probability that more than $s$ tasks are completed in one time interval:
$$\mathbf{Pr}(\mathbf{Pois}(\cdot|\lambda) \geq s) = \sum_{k \geq s} e^{-\lambda} \frac{\lambda^k}{k!} \leq e^{-\lambda} \frac{\lambda^s}{s!} \frac{s}{s - \lambda}$$
becomes negligible, and thus the contribution of those terms in DP update formulas will also become negligible.

In practice, we could set a threshold $\epsilon$ for the probability $\mathbf{Pr}$ $(\mathbf{Pois}(\cdot|\lambda)$ $\geq s)$. If for some $s_0$, $\mathbf{Pr}(\mathbf{Pois}(\cdot|\lambda) \geq s_0)$ is less than the threshold $\epsilon$, all the terms $s > s_0$ can be ignored safely. Table~\ref{tbl_safe_num_task_ub} shows the value of $s_0$ for $\epsilon = 10^{-9}$ and different values of $\lambda$.
\begin{table}[h]
\centering
\scriptsize
\begin{tabular}{|c|c|c|}
\hline
\textbf{Threshold} $\epsilon$ & \textbf{Poisson mean} $\lambda$ & $s_0$\\
\hline 
$10^{-9}$ & 10 & 35\\
\hline 
$10^{-9}$ & 20 & 53\\
\hline 
$10^{-9}$ & 50 & 99\\
\hline
\end{tabular}
\caption{The value of $s_0$ for different thresholds $\epsilon$ and Poisson distribution means}
\label{tbl_safe_num_task_ub}
\vspace{-10pt}
\end{table}

The next theorem provides an upper bound of error produced by Poisson Distribution Truncation:
\begin{theorem}
The exact optimal total cost $\mathbf{Opt}(n, t)$ and estimated value of optimal total cost $\mathbf{Est}_{trunc}(n, t)$ using Poisson Distribution Truncation and the exact total cost $\mathbf{Cost}_{trunc}(n, t)$ based on the optimal policy obtained using Poisson Distribution Truncation satisfies the following inequality:
\begin{align*}
\mathbf{Est}_{trunc}(n, t) \leq \mathbf{Opt}(n, t) \leq \mathbf{Cost}_{trunc}(n, t)\\
\leq \mathbf{Est}_{trunc}(n, t) + \epsilon n (N_T - t) C
\end{align*}
where $C$ is the upper bound of task reward in any state. In particular, 
$|\mathbf{Opt}(N, 0) - \mathbf{Cost}_{trunc}(N, 0)| \leq \epsilon N N_T C$
\end{theorem}
\papertext{The proof of this and subsequent theorems are omitted due to space considerations and can be found in our technical report~\cite{treport}.}
\techreporttext{
\begin{proof}
The former inequalities can be proved by induction in a very straight-forward manner. The last inequality involving state $(N, 0)$ is direct implication of former inequalities.
\end{proof}
}

\stitle{Monotonicity of Pricing Decision:}
Another speed-up technique relies on the following natural conjecture:
\newtheorem{conjecture}{Conjecture}
\begin{conjecture}
The optimal reward $\mathbf{Price}(n, t)$ for each task is non-decreasing with respect to $n$ for any fixed value of $t$.
\end{conjecture}
Intuitively, this conjecture says that with a fixed deadline, the more remaining tasks we have, the higher reward we should set for each task. Over repeated trials with many different values of $\lambda, N, N_T$, we tried generating optimal strategies (using the basic DP algorithm described in the previous section), and the optimal strategies never violate the preceding conjecture.

If we assume this conjecture to be correct, then the following can be used to speed up the DP process. The main idea is to reduce the search range of optimal reward $c$ for each state: suppose $\mathbf{Price}(a, t)$ and $\mathbf{Price}(c, t)$ are already known, then for any $a < b < c$, $\mathbf{Price}(b, t)$ lies in range $[\mathbf{Price}(a, t), \mathbf{Price}(c, t)]$. 
\techreporttext{Figure~\ref{fig_dp2_illustrate} illustrates the idea of this algorithm.} 
For time interval $t$, we first search for the optimal reward for state $(\frac{N}{2}, t)$, then states $(\frac{N}{4},t)$ and $(\frac{3N}{4}, t)$, then states $(\frac{kN}{8}, t)$ for $k = 1, 3, 5, 7$. This process continues until the optimal reward for every state has been found. Thus, the  optimal reward searching process can be represented using a binary tree, where each node represents the optimal reward search range of certain state, and the search range of optimal reward is bounded by optimal reward already found in upper level nodes. 
Further, the search range of nodes in each level sum up to $C$, the pre-specified upper bound of task reward, while the number of levels is bounded by $O(\log n)$.
\papertext{Our technical report~\cite{treport} has a diagram depicting these relationships.}
Therefore, the algorithm 
\techreporttext{(Algorithm~\ref{alg_efficient_dp})}
\papertext{(which can be found in the extended technical report~\cite{treport})} has a time complexity of $O(N_TN(N + C \log N))$. 

Finally, although not improving time complexity, the monotonicity of task rewards $\mathbf{Price}(n, t)$ with respect to $t$ for fixed $n$ (i.e.,  when the number of remaining tasks are fixed, the rewards increase as we get closer to the deadline), can also be used to improve algorithm efficiency by reducing the optimal reward search range.
\techreporttext{
\begin{figure*}
\vspace{-10pt}
\centering
\includegraphics[width = 4in]{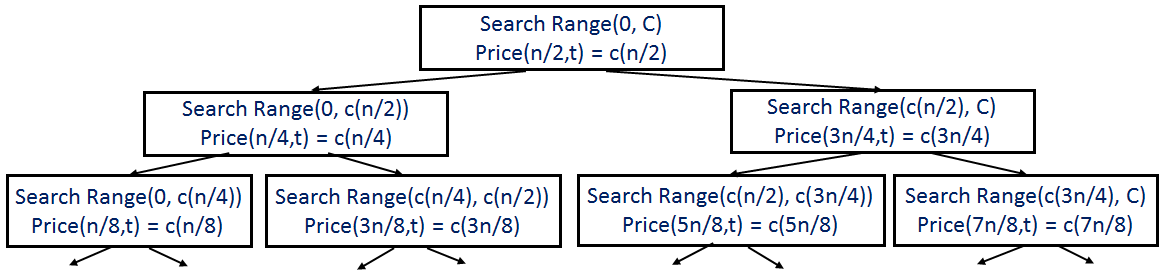}
\caption{The graphical illustration of the efficient algorithm \techreporttext{(Algorithm~\ref{alg_efficient_dp})}, states are represented as nodes in the tree, the search range of each node can be bounded by the optimal price of nodes with lower depth}
\vspace{-10pt}
\label{fig_dp2_illustrate}
\end{figure*}
}
\techreporttext{
\begin{algorithm}[h]
\scriptsize
\caption{Efficient Dynamic Programming}
\label{alg_efficient_dp}
\begin{algorithmic}
\Function{FindOptimalPriceForTime}{$t$, $l$, $r$, $L$, $R$}
	\State $m \leftarrow \lfloor \frac{l + r}{2} \rfloor$
	\State FindOptimalPriceForState($m$, $t$, $L$, $R$)
	\State $pm \leftarrow Price(m, t)$
	\If {$l < m$}
		\State FindOptimalPriceForTime($t$, $l$, $m - 1$, $L$, $pm$)
	\EndIf
	\If {$m < r$}
		\State FindOptimalPriceForTime($t$, $m + 1$, $r$, $pm$, $R$)
	\EndIf	
\EndFunction

\Function{ImprovedDP}{}
	\For{$i = 0$ to $N$}
		\State $Opt(i, N_T) \leftarrow i \times \mathbf{Penalty}$
	\EndFor

	\For{$t = N_T - 1$ to $0$}
		\State FindOptimalPriceForTime($t$, $0$, $N$, $0$, $C$)
	\EndFor
\EndFunction
\end{algorithmic}
\end{algorithm}
}

\techreporttext{
\subsection{Final State Penalties}\label{sec_objectives}

In our MDP formulation, the penalties for final states $\mathbf{cost}\{(n, N_T)\}$ are proportional to the number of remaining tasks left unsolved. Therefore, the MDP is optimizing the linear combination of the total reward paid for the tasks completed before deadline and the number of remaining tasks after deadline:
\begin{align*}
\mathbf{Q}  = \mathbb{E}(\textbf{transition cost}) 
	 + \mathbb{E}(\textbf{\# of unsolved tasks}) \times \mathbf{Penalty}
\end{align*}

The parameter $\mathbf{Penalty}$ controls the trade-off between these two quantities: higher value of $\mathbf{Penalty}$ results in higher average reward for each tasks and less remaining tasks after deadline on average. 

Sometimes, it may be more convenient to directly optimize the expected total expenditure on crowdsourcing marketplace, with a constraint on the expected remaining uncompleted tasks after deadline.
\begin{align*}
\textit{Minimize } &  \mathbb{E}(\textbf{transition cost}) \\
\textit{s.t. } &  \mathbb{E}(\textbf{\# of remaining tasks}) \leq \mathbf{Bound}
\end{align*}

Theorem~\ref{thm_equiv} shows that two formulations are closely related.  
\begin{theorem}\label{thm_equiv}
For every value of parameter $\mathbf{Penalty}$, there exists a corresponding value of parameter $\mathbf{Bound}$ such that two formulations above result in the same optimal solution.
\end{theorem}

\begin{proof}
For any fixed value of $\mathbf{Penalty}$ parameter, assume the optimal solution for the original MDP formulation is $\mathbf{Opt}$. Let $\mathbf{Bound}$ to be the expected number of unsolved tasks in $\mathbf{Opt}$. For any other solution $\mathbf{Sol}$, if the expected number of unsolved tasks in $\mathbf{Sol}$ is less than or equal to $\mathbf{Bound}$, then the expected transition cost of $\mathbf{Sol}$ must be no less than $\mathbf{Opt}$'s (Otherwise the optimality of $\mathbf{Opt}$ is violated). Therefore, $\mathbf{Opt}$ is also optimal in the second formulation.
\end{proof}

Therefore, for any fixed value of parameter $\mathbf{Bound}$, we can perform binary search for the value of parameter $\mathbf{Penalty}$ such that the solution to former formulation is also a solution to latter formulation.

The original final state penalty could extended as follows:
\begin{align*}
\mathbf{cost}\{(n, N_T)\} = \left\{ 
\begin{array}{lr}
(n + \alpha) \times \mathbf{Penalty} & \text{if } n > 0\\
0 & \text{if } n = 0
\end{array}
\right.
\end{align*}
which enforces an extra penalty on the existence of remaining tasks. This formulation may be more suitable for cases where any remaining task would be problematic but the number of remaining tasks does not really matter. Just like the scenario above, there is a correspondence between MDP and the following formulation.
\begin{align*}
\textit{Minimize } & \mathbb{E}(\textbf{transition cost})\\
\textit{s.t. } & \mathbb{E}(\textbf{\# of remaining tasks}) ~ + \\ & \alpha \times \mathbf{Pr}(\textbf{\# of remaining tasks} > 0) \leq \mathbf{Bound}
\end{align*}
Thus, the extended penalty setting would not only bound the average number of unsolved tasks but also bound the probability that there exists at least one remaining task.
}


\section{Fixed Budget Pricing Strategy}\label{sec_fixed_budget}

In this section we focus on another version of pricing problem: given a total monetary budget for all tasks, our objective is to minimize the expected total time when all tasks are completed.

Although, like in the previous section, we may still change the task reward dynamically, we will demonstrate that exercising this freedom does not help much in this scenario. In fact, we will prove that a static pricing strategy is nearly optimal. 

\subsection{Static Pricing Strategy}\label{sec_static_strategy}

We first define what we mean to be a \textbf{Static Pricing Strategy}:
\newdef{definition}{Definition}
\begin{definition}
A \textit{static pricing strategy} assigns a reward to each of the $N$ tasks up-front (i.e., at the time the tasks are submitted to the marketplace), and then does not change this price subsequently. Note that the rewards need not be the same for all tasks. 
\end{definition}


Even though for a static pricing strategy tasks are submitted to the marketplace at the beginning with possibly different rewards, at any time, only the tasks with the highest reward will be picked up by workers. Thus,  the rate at which tasks are picked up by workers will depend solely on the highest reward among all tasks (This property can be shown by \textit{Utility Theory} in Section~\ref{sec_prob_est}). Later on, when the tasks with the highest reward are exhausted, workers will start to pick up tasks with a lower reward; as a result the task acceptance rate will drop accordingly.

Note that static pricing strategies are a strict restriction of general dynamic pricing strategies. To see this, observe that for every static pricing strategy, there is an equivalent dynamic pricing strategy which changes the task reward for all tasks right after each task is completed. Therefore, the optimal static pricing strategy cannot have a lower total latency than the optimal dynamic pricing strategy. However we will show that in fact, the former can have as low expected total latency as the latter.

\subsection{Optimality of Static Pricing Strategy}\label{sec_optimality}

We now show that the optimal static pricing strategy has the minimum expected total latency for completing a given batch of tasks among all possible pricing strategies. Our main result will be Theorem~\ref{thm_static_opt}, described in Section~\ref{sec_worker_quantity}. Subsequent sections will focus on the proof and describe the algorithms. 

\subsubsection{Worker-Arrival Quantity}\label{sec_worker_quantity}

Recall that from Section~\ref{sec_model}, the workers arrive at the marketplace following a NHPP, and decide whether to work on our task following an independent Bernoulli process. Let $T$ be the random variable denoting the total time elapsed before all tasks are completed, and $W$ be the random variable denoting the total number of workers that have arrived at the marketplace before all the tasks are completed. Based on our model, the distribution of $T$ conditioned on $W$ depends only on the arrival-rate parameter $\lambda(t)$, and is independent of the pricing strategy. Suppose we use a pricing strategy $S$, then the expected value of $T$ can be expressed as
\techreporttext{:$$\mathbb{E}[T|S] = \int_W \mathbb{E}[T|W] \mathbf{Pr}(W|S) dW$$}
\papertext{$\mathbb{E}[T|S] = \int_W \mathbb{E}[T|W] \mathbf{Pr}(W|S) dW$.}

Therefore our goal is to choose the optimal pricing strategy such that its induced distribution $\mathbf{Pr}(W|S)$ minimizes $\mathbb{E}[T|S]$. Now if $E[T|W]$ is linear in $W$, then we have:
$$\mathbb{E}[T|S] = \int_W k W \mathbf{Pr}(W|S) dW = k \mathbb{E}[W|S]$$
which means that minimizing $\mathbb{E}[T|S]$ is equivalent to minimizing $\mathbb{E}[W|S]$. Minimizing the latter quantity is much more straight-forward as we will show in next few sections. 
\papertext{The justification for this linearity assumption is omitted due to space limitations, and can be found in the extended technical report~\cite{treport}.}
\techreporttext{The justification of this linearity assumption will be shown in Section~\ref{sec_linear_assumption}.}

The next theorem states that static pricing strategy is optimal in terms of minimizing the expected number of worker-arrivals $\mathbb{E}[W|S]$ and therefore expected latency $\mathbb{E}[T|S]$. We will prove the theorem in the next section. 

\begin{theorem}\label{thm_static_opt}
There exists a static pricing strategy $S$ that minimizes the expected number of total worker-arrivals $\mathbb{E}[W|S]$, and therefore minimizes the expected total latency $\mathbb{E}[T|S]$ among all possible pricing strategies.
\end{theorem}

\techreporttext{
\subsubsection{Linearity Assumption Justification}\label{sec_linear_assumption}
In this section we justify the linearity assumption that $E[T|W] = kW$. First notice that $T$ has the following conditional distribution function conditioned on $W$:
$$F_{T|W}(t) = \mathbf{Pr}(T \leq t|W) = \mathbf{Pr}(N(t) \geq W)$$
where $N(t)$ is the random variable denoting the number of workers who have arrived at the marketplace between time $0$ and time $t$. Based on the NHPP model, $N(t)$ follows Poisson distribution:
\begin{equation}\label{eqn_NHPP}
N(T) \sim \mathbf{Pois}(\cdot|\lambda = \int_0^T \lambda(t) dt)
\end{equation}

As shown in Figure~\ref{fig_mturk-track_hit_completion}, $\lambda(t)$ varies periodically and is relatively stable over a long period. Thus $\int_0^T \lambda(t)$ is approximately proportional to $T$:
$$\Lambda(T) = \int_0^T \lambda(t) dt \approx \bar \lambda T$$
where $\bar \lambda$ is the average worker-arrival rate in the marketplace. On substituting it into Equation~(\ref{eqn_NHPP}) we have,
\begin{align*}
\mathbf{Pr}(N(t) \geq W) & \approx 1 - \sum_{k=0}^{W-1} \mathbf{Pois}(k|\lambda = \bar \lambda T)\\
& = 1 - e^{-\bar \lambda T} \sum_{k=0}^{W-1} \frac{(\bar \lambda T)^k}{k!}
\end{align*}
Therefore\footnote{We have used a fact in probability theory that for any non-negative random variable $X$, $\mathbb{E}(X) = \int_0^\infty \mathbf{Pr}(X > t) dt$},
\begin{align*}
\mathbb{E}(T|W) & = \int_0^\infty (1 - F_{T_W}(t)) dt \approx \int_0^\infty e^{-\bar \lambda t} \sum_{k=0}^{W-1} \frac{(\bar \lambda t)^k}{k!} dt\\
& = \sum_{k=0}^{W-1} \int_0^\infty e^{-\bar \lambda t} \frac{(\bar \lambda t)^k}{k!} dt = \frac{W}{\bar \lambda}
\end{align*}
which justifies that linearity assumption.
}

\subsubsection{Optimality of Static Pricing Strategy}\label{sec_semi-static}

The proof of Theorem~\ref{thm_static_opt} relies on another type of pricing strategy: \textit{Semi-Static Pricing Strategy}. Semi-static pricing strategies serve as a bridge to connect static pricing strategies and dynamic pricing strategies in the proof of Theorem~\ref{thm_static_opt}:
\begin{definition}
A \textit{Semi-Static Pricing Strategy} generates a sequence of prices $c_1, c_2, \ldots, c_N$ at the time the tasks are posted to the marketplace. 
The strategy starts off by assigning $c_1$ to all tasks, and once one task is picked up by a worker, the price for all remaining tasks changes to $c_2$, and so on, until all the tasks are picked up by workers and completed. Unlike the static pricing strategy, the sequence of $c_i$'s need not be monotonically decreasing. 
\end{definition}


We next show that the best dynamic pricing strategy is as good (i.e., has as low an expected completion time or latency) as the best semi-static pricing strategy. 
\begin{theorem}\label{thm_opt_semi}
The optimal dynamic pricing strategy to minimize the expected number of worker-arrivals $\mathbb{E}[W]$ is in the form of a semi-static pricing strategy.
\end{theorem}
\papertext{The proof of this theorem can be found in the extended technical report~\cite{treport}.} Intuitively, the proof uses decision theory to demonstrate that, for a dynamic strategy, only the decisions made when a task gets completed matter --- otherwise the state of the Markov process stays the same, and need not be changed.
\techreporttext{
\begin{proof}
The optimal dynamic pricing strategy that minimizes the expected number of worker-arrivals $\mathbb{E}[W]$ can be obtained by solving the corresponding Markov Decision Process. Since we are considering the number of worker-arrivals as cost, the MDP can be represented by tuple $(n, b)$ representing the number of remaining tasks and total budget left, and the transition between states are:
\begin{align*}
\mathbf{Pr}\{(n, b) \rightarrow (n-1, b - c)\} = p(c)\\
\mathbf{Pr}\{(n, b) \rightarrow (n, b)\} = 1 - p(c)
\end{align*}
with corresponding cost:
\begin{align*}
\mathbf{Cost}\{(n, b) \rightarrow (n-1, b - c)\} = 1\\
\mathbf{Cost}\{(n, b) \rightarrow (n, b)\} = 1
\end{align*}
where the two state transitions above indicates whether the next arrived worker will accept the task if reward is $c$. 

A special property of this MDP is that each state has only one outgoing transition edge, which corresponds to the event that some worker accepts our task and completes it. Therefore, the MDP formulation indicates that the task reward will remain unchanged in the optimal pricing strategy until some task is completed (since otherwise the state is still the same). In other words, the optimal pricing strategy is in the form of semi-static pricing strategy.
\end{proof}
}

The next theorem states that the effectiveness of any semi-static pricing strategy is not affected by the order of the $c_i$. 
\begin{theorem}\label{thm_eff_semi}
For any semi-static pricing strategy $S$ with price sequence $c_1, c_2, \ldots, c_N$, then the expected number of worker-arrivals $\mathbb{E}[W]$ is equal to $\sum_{i=1}^N \frac{1}{p(c_i)}$.
\end{theorem}
\papertext{The proof of this theorem can also be found in the extended technical report~\cite{treport} as well. The proof sketch is as follows: We decompose $W$ into sum of $n$ independent random variables each representing the number of worker-arrivals between successive tasks. In doing so, we can rewrite the expected value of $W$ as the sum of expected value of each of those $n$ random variables. Finally the expected value of each individual random variable can be calculated directly.}

\techreporttext{
\begin{proof}
Let $w_i$ denotes the number of worker-arrivals between the completion time of $(i - 1)$th and $i$th task ($w_1$ denotes the number of worker-arrivals before the completion time of the $1$st task). Based on the model assumption in Section~\ref{sec_model}, we can derive that $w_i$ follows geometric distribution:
$$\mathbf{Pr}[w_i = k] = (1 - p(c_i))^k p(c_i)$$
where $p(c_i)$ is the task acceptance probability with respect to task reward $c_i$. 

The total number of worker-arrivals $W$ can be expressed as sum of $w_i$s plus $N$ workers that actually picked up tasks: 
$$W = \sum_{i=1}^N w_i + N$$
Taking expectation on both side, we get:
$$\mathbb{E}[W] = \sum_{i=1}^N \mathbb{E}[w_i] + N  = \sum_{i=1}^N \frac{1 - p(c_i)}{p(c_i)} + N = \sum_{i=1}^N \frac{1}{p(c_i)}$$
which finishes the proof.
\end{proof}
}
Thus, by reordering the prices of a semi-static strategy (to ensure a descending order), we can change it into a static strategy with equal total expected completion time or latency. This result together with Theorem~\ref{thm_opt_semi} demonstrates that the static pricing strategies are near-optimal.

\subsection{Nearly Optimal Solution via LPs} \label{sec_lps}

In this section, we will address the problem of finding the optimal static pricing strategy. Suppose in the optimal static pricing strategy, the rewards for tasks are $c_1, c_2, \ldots, c_N$. Using Theorem~\ref{thm_eff_semi}, we know that the expected total number of worker-arrivals $\mathbb{E}[W]$ equals the sum of $\frac{1}{p(c_i)}$ (since any static pricing strategy is also a semi-static pricing strategy with the reward sequence monotonically non-increasing):
\begin{equation}\label{eqn_seq}
\mathbb{E}[W] = \sum_{i=1}^N \frac{1}{p(c_i)}
\end{equation}

Let $n_c$ be the number of tasks with reward $c$, i.e., $n_c = |\{i: c_i = c\}|$.
Then, Equation~(\ref{eqn_seq}) can be rewritten as: $\mathbb{E}[W] = \sum_c n_c \frac{1}{p(c)}$. The $n_c$ values satisfy the following constraints:
\begin{equation}\label{eqn_constrain}
\sum_c n_c = N; \ \ \sum_c n_c \times c \leq B; \ \ n_c \geq 0; \ \ n_c \in \mathbb{N}
\end{equation}
where the first constraint is about the total number of tasks, the second constraint is about the total monetary budget ($B$ denotes the total budget for all tasks).

Our objective is to find values of $n_c$ that minimizes $\mathbb{E}[W]$ while simultaneously satisfying Constraints~(\ref{eqn_constrain}). 
For arbitrary functions $p(c)$, it is easy to show that the optimization problem is {\sc NP-Hard}. Furthermore, we can show that the optimal static pricing strategy solution can be generated using a dynamic-programming based pseudo-polynomial time algorithm:
\begin{theorem}
The $c_i$ for the optimal static pricing strategy can be discovered in {\sc Ptime} $(B, N)$.
\end{theorem}
In short, the idea is to consider all optimal allocations of up to $B$ to the first $i$ tasks, for all $i \in 1\ldots n$.
 
Our approach will instead be to approximately solve the optimization problem. We begin by casting the problem as an Integer Program (IP). Then, we will relax the IP to a Linear Program (LP) where the variables no longer have to be integers, i.e., the $n_c \in \mathbb{N}$ constraints are excluded.  Lastly, we will round up variables in the solution to the LP to make them integers. The relaxed LP version of the problem is as follows:
\begin{align*}
\textit{Minimize }  \sum_c n_c \frac{1}{p(c)} 
\textit{ \ \ s.t. } \sum_c n_c = N; \sum_c n_c \times c \leq B; n_c \geq 0
\end{align*}
Instead of applying an LP solver and then performing rounding, next, we will describe an even faster approach, that leverages a special property of the LP above:
\begin{theorem}\label{thm_lp_properties}
There exists an optimal solution for the LP above which satisfies the following:
\squishlist
\item
$\exists c_1 < c_2$, $\forall c \not= c_1, c \not= c_2, n_c = 0$
\item
$\forall c = t c_1 + (1 - t) c_2, t \in \mathbb{R}$:
$\frac{1}{p(c)} \geq t \frac{1}{p(c_1)} + (1 - t) \frac{1}{p(c_2)}$
\squishend
\end{theorem}

Theorem~\ref{thm_lp_properties} can be intuitively explained using Figure~\ref{fig_lp_illustrate}. We first plot all the pairs $(c, \frac{1}{p(c)})$ in the plane. 
The first property of Theorem~\ref{thm_lp_properties} states that there are at most two $c_i$'s with non-zero $n_{c_i}$; i.e., there are at most two distinct prices $c_1, c_2$ that tasks are set at.
Then, the second property of Theorem~\ref{thm_lp_properties} states that for $c_1$ and $c_2$, there is no other point $(c, \frac{1}{p(c)})$ below the straight line connecting $(c_1, \frac{1}{p(c_1)})$ and $(c_2, \frac{1}{p(c_2)})$. In other words, $(c_1, \frac{1}{p(c_1)})$ and $(c_2, \frac{1}{p(c_2)})$ can only be segments on the \textit{convex hull} of points $(c_i, \frac{1}{p(c_i)})$.
\begin{figure}
\centering
\vspace{-5pt}
\includegraphics[width = 1.8in]{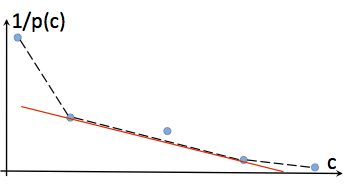}
\caption{Illustration of Theorem~\ref{thm_lp_properties}, which implies that $c_1$ and $c_2$ can only be on the convex hull}
\label{fig_lp_illustrate}
\vspace{-10pt}
\end{figure}

\papertext{The proof of Theorem~\ref{thm_lp_properties} can be found in the extended technical report~\cite{treport}.} The key idea is to show that given any optimal solution, it can be transformed to satisfy the first property while maintaining its optimality. The second property can be derived from the first property and Karush-Kuhn-Tucker conditions~\cite{boyd2004convex} of the LP.
\techreporttext{
\begin{proof}[of Theorem~\ref{thm_lp_properties}]
Suppose $n^*$ is the optimal solution to the above LP problem. Let $c_1 = \min\{c: n^*_c > 0\}$ and $c_2 = \max\{c: n^*_c > 0\}$ be the smallest and biggest index of non-zero component of $n^*$ respectively. We show that the following solution is also an optimal solution:
\begin{align*}
n'_{c_1} = \sum_{c_1 \leq c \leq c_2} n^*_c \frac{c_2 - c}{c_2 - c_1}\\
n'_{c_2} = \sum_{c_1 \leq c \leq c_2} n^*_c \frac{c - c_1}{c_2 - c_1}\\
\forall c \not= c_1, c \not= c_2, n'_c = 0
\end{align*}
In order to prove this claim, we need to show that:
$$\sum_c n^*_c \frac{c_2 - c}{c_2 - c_1} \frac{1}{p(c_1)} + \sum_c n^*_c \frac{c - c_1}{c_2 - c_1} \frac{1}{p(c_2)} \leq \sum_c n^*_c \frac{1}{p(c)}$$
It suffices to prove that:
\begin{equation}\label{eqn_kkt_cond}
\forall c, \frac{1}{p(c)} \geq \frac{c_2 - c}{(c_2 - c_1)p(c_1)} + \frac{c - c_1}{(c_2 - c_1)p(c_2)}
\end{equation}
In order to prove Equation~\ref{eqn_kkt_cond}, we examine the Karush-Kuhn-Tucker conditions~\cite{boyd2004convex} of this LP:
$$\forall c, \frac{1}{p(c)} = \mu_c + \lambda_N - c \mu_B, \mu_c \geq 0, \mu_c n^*_c = 0$$
where $\mu_c, \lambda_N, \mu_B$ are KKT multipliers. Since $\mu_c \geq 0$, it follows
\begin{equation}\label{eqn_kkt_imp}
\forall c, \frac{1}{p(c)} \geq \lambda_N - c \mu_B
\end{equation}
with equality holds on $c = c_1$ and $c = c_2$ (since $n^*_{c_1}, n^*_{c_2} > 0$ implies $\mu_{c_1} = \mu_{c_2} = 0$). Substitute Equation~(\ref{eqn_kkt_imp}) into Equation~(\ref{eqn_kkt_cond}) completes the proof of the first part. The second claim is a direct implication of Equation~(\ref{eqn_kkt_imp}) with equality holds on $c=c_1, c_2$.
\end{proof}
}

Using Theorem~\ref{thm_lp_properties}, we can derive an algorithm (Algorithm~\ref{alg_static_strategy}) to find a nearly optimal pricing strategy. The algorithm generates the convex hull using all possible prices, and then picks the two most suitable prices to assign to tasks. 

Theorem~\ref{thm_lp_guarantee} provides an upper bound of the difference between rounded-LP solution (i.e., the solution provided by Algorithm~\ref{alg_static_strategy}) and optimal solution of original IP problem.

\begin{algorithm}
\scriptsize
\caption{Find Optimal Static Pricing Strategy}
\label{alg_static_strategy}
\begin{algorithmic}
\Function{FindOptimalStaticStrategy}{}
	\For{$c = 0$ to $C$}
		\State Calculate the value of task acceptance probability $p(c)$.
	\EndFor
	\State $CH \leftarrow \textbf{Convex hull of points } (c, \frac{1}{p(c)})$
	\State $c_1 \leftarrow \max\{c \in CH : c \leq \frac{B}{N}\}$
	\State $c_2 \leftarrow \min\{c \in CH : c > \frac{B}{N}\}$
	\State $n_1 \leftarrow \lceil \frac{c_2 N - B}{c_2 - c_1} \rceil$, $n_2 \leftarrow N - n_1$
	\State \Return $n_1$ tasks priced at reward $c_1$; $n_2$ tasks at reward $c_2$.
	\EndFunction
\end{algorithmic}
\end{algorithm}
\begin{theorem}\label{thm_lp_guarantee}
Let $\{n^*\}$ denote the optimal solution that minimizes $\mathbb{E}[W]$ under the Constraint~(\ref{eqn_constrain}), and $\{\hat{n}\}$ denote the rounded-LP solution from Algorithm~\ref{alg_static_strategy}, then the expected total latency difference between two solutions is bounded by:
{\scriptsize
$$ \sum_c \hat{n}_c \frac{1}{p(c)} \leq \sum_c n^*_c \frac{1}{p(c)} + (\frac{1}{p(c_1)} - \frac{1}{p(c_2)})$$
}
\end{theorem}
\techreporttext{
\begin{proof}
Since relaxation only removed integer restriction (without adding any constraints), it implies that $\{n^*\}$ is also a valid solution to the relaxed LP problem. Therefore, the optimal LP solution $n^{L}$ will achieve lower objective function value than $\{n^*\}$:
$$ \sum_c n^{L}_c \frac{1}{p(c)} \leq \sum_c n^*_c \frac{1}{p(c)} $$
Since $\hat{n}$ is just the rounded-solution of $n^L$, together with the special form of $n^L$ implied by Theorem~\ref{thm_lp_properties}, we get:
$$ \sum_c \hat{n}_c \frac{1}{p(c)} \leq \sum_c n^{L}_c \frac{1}{p(c)} + (\frac{1}{p(c_1)} - \frac{1}{p(c_2)})$$
Combining two results completes the proof.
\end{proof}
}


\section{Experiments}\label{sec_experiments}

The goals of our experimental evaluation are two-fold: (a) to validate the pricing model assumptions we made in the previous sections, and (b) to compare our techniques versus others on simulations based on real crowdsourcing marketplace data, as well as real experiments deployed on a 
crowdsourcing marketplace.
In Section~\ref{sec_exp_tap}, we examine the validity of the task acceptance probability equation (Equation~(\ref{eqn_accept_prob})) and estimate the typical task acceptance probability  values for real tasks. 
In Section~\ref{sec_exp_fixed_deadline}, we examine the effectiveness (in terms of total monetary cost) of our techniques for the fixed deadline problem from Section~\ref{sec_fix_deadline} as compared to other schemes under simulations based on real workloads from Amazon's Mechanical Turk Marketplace via the mturk-tracker website~\cite{mturk-tracker}. We also study the sensitivity of our techniques with respect to (a) the algorithm parameters, (b) the estimation error of arrival-rate, and (c) the task acceptance probability mapping function, since many of these parameters may only be estimated approximately. 
In Section~\ref{sec_exp_mturk}, we deploy our pricing technique for the fixed deadline problem from Section~\ref{sec_fix_deadline} on Mechanical Turk and report effectiveness in practice. 
(In \techreporttext{Section~\ref{exp_mturk_data_analysis}, we}\papertext{the extended technical report, we} present some data analysis of the data collected as a result.)
In \techreporttext{Section~\ref{sec_exp_static}, we}\papertext{the extended technical report, we additionally} examine the completion times of our techniques for the fixed budget problem from Section~\ref{sec_fixed_budget} under simulations based on real workloads.

\subsection{Task Acceptance Probability}\label{sec_exp_tap}

In Section~\ref{sec_prob_est}, we used Equation~(\ref{eqn_accept_prob}) to map task rewards to task acceptance probabilities. In this section, we experimentally validate Equation~(\ref{eqn_accept_prob}) 
using utility theory (Section~\ref{sec_exp_utility}) and estimate the parameters in Equation~(\ref{eqn_simplified}) for real tasks (Section~\ref{sec_exp_coefficient}).

\subsubsection{Utitity-based Simulation}\label{sec_exp_utility}

As described in Section~\ref{sec_prob_est}, workers choose tasks to work on by maximizing their gain in utility. In this section we simulate a specific workers' choice based on utility theory to justify the form of Equation~(\ref{eqn_accept_prob}). 

The experiment settings are the following:
\begin{denselist}
\item
The total number of tasks on the marketplace is set to be $100$.
\item
The worker's utility estimate $U_i$ for task $T_i(i > 1)$ follows a normal distribution $\mathcal{N}(\mu_i, \sigma_i^2)$, where $\mu_i$ are sampled independently from the normal distribution $\mathcal{N}(0, 1)$, and $\sigma_i$ are sampled independently from the uniform distribution $U[0, 1]$.
\item
The worker's utility estimate $U_1$ for our target task $T_1$ follows a normal distribution $\mathcal{N}(\mu_1 = \frac{c}{50} - 1, \sigma_1^2)$ where $c$ denotes the task reward of our task $T_1$ and $\sigma_1$ is sampled from the uniform distribution $U[0, 1]$. 
\end{denselist}
For a given $c$ (i.e., the reward for our task), we repeatedly sample the utility estimates for each of the $100$ tasks as described above, and assume that the worker will choose our task if and only if our task has the highest utility among all the tasks in the marketplace. 
This sampling process gives us an estimate of the task acceptance probability $p$ for a fixed reward $c$.
We then repeat this process for different values of $c$, and plot the
simulated acceptance probability $p$ over different values of $c$ in 
Figure~\ref{fig_acc_synthetic}. 
In the figure, we also depict the corresponding regression curve based on Equation~(\ref{eqn_accept_prob}) for comparison (the value of $\beta$ is learned by fitting the simulated task acceptance probability value). 
As can be seen in Figure~\ref{fig_acc_synthetic}, the simulated acceptance probability $p$ is well predicted by Equation~(\ref{eqn_accept_prob}). This justifies the model assumption that $p$ is proportional to the exponential of the task utility $U_i$. 

\begin{figure}[h]
\vspace{-5pt}
\centering
\includegraphics[width = 2in]{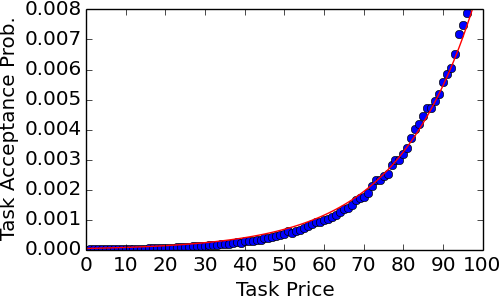}
\vspace{-5pt}
\caption{Simulated task acceptance probability $p$ with reward $c$ ranging from $0$ to $100$. Blue dots are simulation results and red curve is the regression function based on Equation~(\ref{eqn_accept_prob}) with $z_i = \mu_i$ and $\beta = 2.6$. }
\label{fig_acc_synthetic}
\vspace{-10pt}
\end{figure}

\subsubsection{Real World Data}\label{sec_exp_coefficient}

In Section~\ref{sec_prob_est}, we use Equation~(\ref{eqn_simplified}) as parametric form of task acceptance probability function.
In this section, we aim to estimate the typical values of parameters $s$, $b$, $M$ in Equation~(\ref{eqn_simplified}) for tasks on a real marketplace.

We retrieved the snapshots of Amazon Mechanical Turk~\cite{mturk} from mturk-tracker~\cite{mturk-tracker}. The snapshots of the marketplace are taken every 20 minutes; we
estimate the number of tasks that are completed every $20$ minutes by subtracting the number of remaining tasks in each task group (note that in Mechanical Turk a task is called a HIT and a group or batch of tasks is called a HIT group).
If the number of remaining tasks increased during the $20$ minute window(i.e., the task requester added new tasks to this task group), we simply assume no tasks was completed during that $20$ minutes.
\eat{However, this estimate may not be perfectly accurate because we find that sometimes the number of remaining tasks has increased in the next snapshot. This sort of situation arises when a task requester has added new tasks to an existing task group in between two snapshots. In that case, the number of tasks completed in that time period will not be available and we simply would assume it to be $0$.} 

We sampled $100$ task groups that had at least $50$ tasks completed (we enforced this threshold to filter out spam tasks) from 1/1/2014---1/28/2014, and for each task group we manually estimated the approximate average time usage for completing one task. Figure~\ref{fig_amt_hits} shows for the two most popular task types, the \textit{wage per second} versus average completed \textit{workload per hour}, defined as:
\begin{align*}
\scriptsize
\textbf{workload per hour} = \textbf{average \# of completed tasks per hour}\\ \scriptsize \times \ \textbf{average time usage of each task}
\end{align*}
We use these two values as axes because we want to make sure that our figure is invariant under task bundling. (In  Mechanical Turk, requesters often group several tasks into one larger task.) 

\begin{figure}[h]
\vspace{-5pt}
\centering
\subfigure[Categorization]{\label{fig_amt_cate} \includegraphics[width = 1.6in]{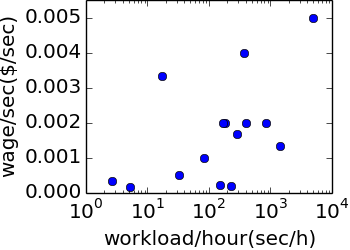}}
\subfigure[Data Collection]{\label{fig_amt_dc} \includegraphics[width = 1.6in]{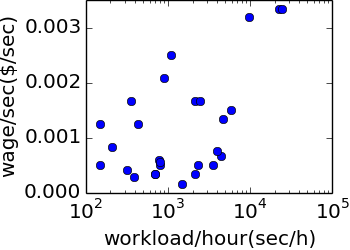}}
\vspace{-10pt}
\caption{The plot of tasks in Amazon Mechanical Turk, x axis represents the wage per second(\$/sec), y axis represents the average completed workload per hour(sec/h) 
}
\label{fig_amt_hits}
\vspace{-10pt}
\end{figure}

In order to estimate the value of parameters $s$, $b$, $M$ in Equation~(\ref{eqn_simplified}), we assume that the utility of each task equals the logarithm of \textit{workload per hour}, as implied by Equation~(\ref{eqn_accept_prob}) if we assume the sum of the exponential of the utilities of all tasks is a constant. We further assume that the utility of each task is linearly correlated with the \textit{wage per second} attribute:
$$\log \textbf{workload/hour} = \textbf{utility} = \alpha \times \textbf{wage/sec} + b + \epsilon$$
where $b$ is task-type bias term and $\epsilon$ accounts for all other factors affecting utility.
We then apply Least Square Regression to estimate the linear coefficient $\alpha$ and the bias term $b$ for each task type. 

Table~\ref{tbl_linear_regression} shows the result of Least Square Regression. The two linear coefficients are approximately the same, implying that the linear coefficient of the \textit{wage per second} attribute is the same for all task types. The bias term of Data Collection tasks is significantly higher than Categorization tasks, implying that workers in Mechanical Turk prefer Data Collection tasks to Categorization tasks.

\begin{table}[h]
\vspace{-5pt}
\scriptsize
\centering
\begin{tabular}{|c|c|c|}
\hline
 & Linear coefficient & Bias\\
\hline 
Categorization & 748 & 3.66\\
\hline 
Data Collection & 809 & 6.28\\
\hline
\end{tabular}
\vspace{-5pt}
\caption{Linear coefficients and bias terms generated using Least Square Regression}
\label{tbl_linear_regression}
\vspace{-10pt}
\end{table}

Using the results in Table~\ref{tbl_linear_regression}, we can then estimate the parameters $s$, $b$, $M$ in Equation~(\ref{eqn_simplified}). Say our task is a Data Collection task and the average completion time of our task is $120$ seconds, then based on Table~\ref{tbl_linear_regression}, we have (task reward $c$ is in cents):
{\scriptsize
\begin{align*}
\textbf{workload per hr.} & = \exp\{809 \times \frac{c}{100} \times \frac{1}{120} + 6.28\}
 = \textbf{total} \times p(c) \times 120
\end{align*}
}
where $\mathbf{total}$ denotes the total number of tasks completed per hour in the crowdsourcing marketplace (including all other tasks). In Mechanical Turk we have $\mathbf{total} \approx 6000$ (as seen in mturk-tracker data). Using this fact we derive the following expression for $p(c)$:
\begin{equation}\label{eq_expt_pc}
p(c) \approx \frac{\exp\{\frac{c}{15} + 0.39\}}{\exp\{\frac{c}{15} + 0.39\} + 2000}
\end{equation}

\subsection{Fixed Deadline Pricing Simulation}\label{sec_exp_fixed_deadline}

In this section, we examine the effectiveness of the dynamic pricing strategy in Section~\ref{sec_fix_deadline}. We compare our dynamic pricing strategy against the binary-search-based fixed pricing strategy in Faridani's work~\cite{faridani11}. We first compare the two pricing strategies under a realistic crowdsourcing workload in Section~\ref{sec_exp_simple}. We then study the trend of relative reduction in cost of our dynamic pricing strategy compared to Faridani's fixed pricing strategy under different problem settings in Section~\ref{sec_exp_trend}. We further examine the sensitivity of dynamic pricing strategy to the granularity of time discretization in Section~\ref{sec_exp_time_interval}. The sensitivity of both pricing strategies to the estimations of task acceptance probability function and future arrival-rates are examined in Section~\ref{sec_exp_para_estimation} and Section~\ref{sec_exp_arrival_rate} respectively.

In the following experiments, we assume the following default settings unless explicitly stated:
\begin{denselist}
\item
The total number of tasks $N = 200$.
\item
The total time before deadline $T = 24 \text{ hours}$.
\item
We retrieved the number of tasks completed during every 20 minutes interval for the time period from 1/1/2014 to 1/28/2014 from mturk-tracker as described in Section~\ref{sec_exp_coefficient}. The worker arrival rate $\lambda(t)$ is set to be piecewise constant on every such $20$ minute time interval, i.e., for each time interval, $\lambda(t)$ is set to match the retrieved arrival data.
\item
Our target task is assumed to be a Data Collection task with an average completion time of $2$ minutes. The mapping function between task reward $c$ and task acceptance probability $p$ can be derived as in Equation~\ref{eq_expt_pc} in 
Section~\ref{sec_exp_coefficient}.
\item
The dynamic pricing model is trained using a time interval of length $20$-minutes.
\end{denselist}

\subsubsection{Effectiveness under a Realistic Workload}\label{sec_exp_simple}

In this section we examine the effectiveness of our dynamic pricing strategy under a realistic workload derived from mturk-tracker. 

We compare our dynamic pricing strategy with fixed price strategies which assigns a fixed reward to all tasks in advance, determined using binary search, and does not change the reward afterwards. Figure~\ref{fig_fixed_deadline} shows the results of our experiment: for various values of the average reward (y axis), we plotted the expected number of tasks that remain unsolved at the deadline (x axis)---the total reward can be estimated by multiplying the average reward with the number of tasks. The $\mathbf{Penalty}$ parameter (see Section~\ref{sec_mdp}) is set in our dynamic pricing strategy such that the expected number of remaining tasks matches those of the fixed pricing strategy.



\begin{figure}[h]
\vspace{-5pt}
\centering
\subfigure{\label{fig_fixed_deadline} \includegraphics[width = 1.6in]{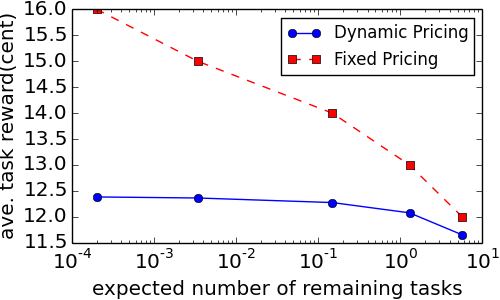}}
\subfigure{\label{fig_trend_N_T} \includegraphics[width = 1.6in]{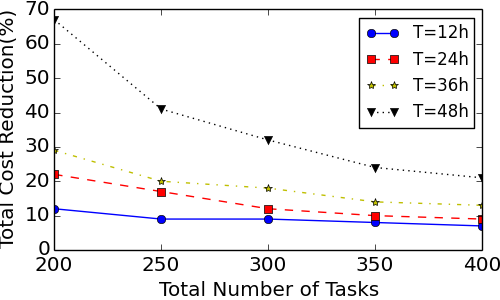}}
\vspace{-5pt}
\caption{(a) Simulated average task reward $c$ of our dynamic pricing strategy with respect to different threshold for the expected number of remaining tasks after deadline. (b) Percentage Cost Reduction with various settings of $N$ and $T$ }
\end{figure}

From Figure~\ref{fig_fixed_deadline}, we see that with low expected number of remaining tasks after deadline (less than $1$ remaining tasks on average), the dynamic pricing strategy achieves an average task reward between $12$ and $12.5$. In fact, we can show that this average reward is very close to the theoretical lower bound of average task reward $c_0$ for any pricing strategy, which satisfies the following equation:
$p(c_0) = \frac{N}{\int_0^T \lambda(t) dt}$.
In our experiment, $c_0 \approx 12$.

The task reward $c_0$ has the following intuitive meaning: Suppose that we have an infinite number of tasks that can be picked up by workers. Let $X$ denotes the number of tasks completed before deadline. Then $c_0$ is the minimum task reward such that: $\mathbb{E}[X] \geq N$.
However, in practice we want to complete all $N$ tasks before deadline, which is equivalent to: $\mathbf{Pr}(X \geq N) \approx 1$.
Note that any pricing strategy satisfying the second constraint will automatically satisfy the first constraint. 
Therefore, in order to achieve a high probability guarantee, the average task reward will necessarily be higher than $c_0$ (since $c_0$ is the minimal possible task reward to satisfy the first constraint of $\mathbb{E}[X] \geq N$). 

Figure~\ref{fig_fixed_deadline} shows that our dynamic pricing strategy can finish all tasks by the deadline with very high probability (99.9\%) and only 3\% overhead (as compared to $c_0$). On the other hand, for the fixed reward pricing strategy from~\cite{faridani11}, the task reward needs to be set at $16$ to achieve the same guarantee, resulting in a 33\% increase over our dynamic pricing strategy, a significant difference in cost.

\subsubsection{Trends of Effectiveness}\label{sec_exp_trend}

In this section, we examine the relative gain of the dynamic pricing strategy compared to a fixed pricing strategy under various settings. We compute the cost reduction achieved by using the dynamic pricing strategy instead of the fixed pricing strategy (in percentage), and study how the reduction changes when the parameters are varied. The experiment settings are listed below:
\begin{denselist}
\item
We study the relative effectiveness of dynamic pricing strategy with respect to the above five parameters of the experiment: $N$, $T$, and three parameters $s$, $b$, $M$ in Equation~(\ref{eqn_simplified}).
\item
Each time we vary only one experiment parameter ($N$, $T$, $s$, $b$, $M$) while keeping other parameters fixed. The default value of these experiment parameters are: $N = 200$, $T = 24$ hours, $s = 15$, $b = -0.39$, $M = 2000$ (same as before). 
\item 
We will compute the total cost of all tasks using both pricing strategies. Let $c_d$ and $c_f$ be the total cost of the dynamic pricing strategy and the fixed pricing strategy respectively, the percentage cost reduction $r$ is defined as:
$r = \frac{c_f - c_d}{c_f}$.
The percentage cost reduction $r$ serves as a measure of the effectiveness of the dynamic pricing strategy as compared to the fixed pricing one. 
\item
For both the dynamic pricing strategy and the fixed pricing strategy, the task reward is chosen such that all tasks are finished by the deadline with $99.9$\% confidence. This will be the default setting for the following experiments.
\end{denselist}
Figure~\ref{fig_trend_N_T} shows the percentage cost reduction under various settings of $N$ and $T$. The experiment shows that the percentage cost reduction decreases as $N$ increases and increases as $T$ increases. Therefore, if we have less number of tasks and time before deadline is longer, then the gain of the dynamic pricing strategy is higher. On the other hand, the gain of the dynamic pricing strategy is lower if we want to complete more tasks in a shorter period of time. The intuitive explanation for this behavior is that with a longer period of time we have the ability to plan ahead and vary the price over time to get additional monetary cost savings. 


Figure~\ref{fig_trend_s_b_M}(a)--(c) shows the trend of percentage cost reduction when the parameter value of $s$, $b$, $M$ changes. The implications can be summarized as follows: 
\begin{denselist}
\item The gain of the dynamic pricing strategy is stable no matter how much the task acceptance probability $p$ is sensitive to task reward $c$ (Figure~\ref{fig_reduction_scale}); 
\item The gain is lower if the task content is intrinsically more attractive compared to other tasks (Figure~\ref{fig_reduction_bias}); 
\item The gain is higher if there are less tasks in the crowdsourcing marketplace than average (Figure~\ref{fig_reduction_total}). 
\end{denselist}

\begin{figure}[h]
\centering
\vspace{-5pt}
\subfigure[Cost Reduction w.r.t. $s$]{\label{fig_reduction_scale} \includegraphics[width = 1.6in]{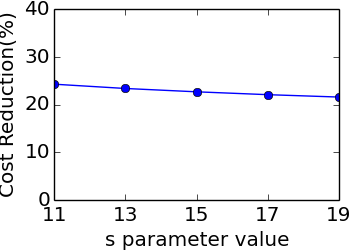}}
\subfigure[Cost Reduction w.r.t. $b$]{\label{fig_reduction_bias} \includegraphics[width = 1.6in]{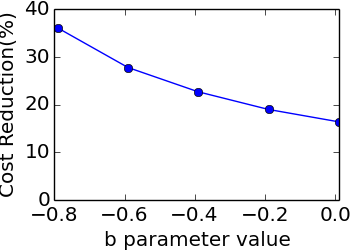}}
\subfigure[Cost Reduction w.r.t. $M$]{\label{fig_reduction_total} \includegraphics[width = 1.4in]{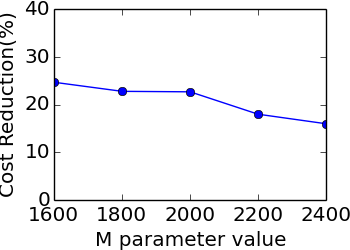}}
\subfigure[Average Task Price for Time Granularities]{\label{fig_granularity} \includegraphics[width = 1.65in]{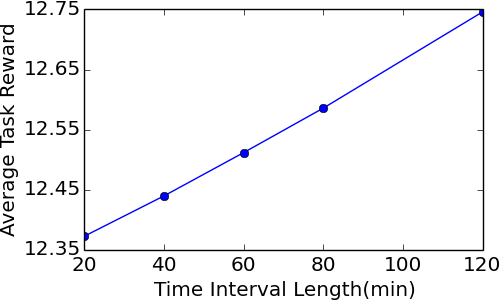}}
\vspace{-10pt}
\caption{(a -- c) Percentage Cost Reduction on varying $s$, $b$, $M$ (d) Task price variation with Granularity}
\label{fig_trend_s_b_M}
\vspace{-12pt}
\end{figure}

\subsubsection{Granularity of Time Interval}\label{sec_exp_time_interval}

In this section, we examine the effects of different time interval granularities. We train our dynamic pricing strategy using different time interval lengths, and examine the corresponding trade-off between effectiveness of pricing strategy and training time. 
The length of time interval used for training the dynamic pricing strategy ranges from $20$ minutes to $2$ hours.

Intuitively, the average task price should increase as the length of time interval increases since the strategy space is reduced; our experiment results in Figure~\ref{fig_granularity} depict the expected behavior: The average task price increases steadily (but not by too much) as the length of time interval increases. On the other hand, the algorithm running time is rather stable and is not affected by the length of time interval (the algorithm running time is between $4$ seconds and $5$ seconds for all experiments, by executing Python code on a laptop with an Intel i7 processor). The stable behavior of running time is probably because of the Poisson truncation technique in Section~\ref{sec_speedup}: the expected number of workers arriving into the marketplace during each time interval will decrease as the length of time interval decreases, and the corresponding Poisson truncation threshold will also decrease. These results argue for using as small a time interval for which we can reliably obtain $\lambda(t)$ data.

\subsubsection{Sensitivity of Parameter Estimation}\label{sec_exp_para_estimation}

Our dynamic pricing strategy (as well as Faridani's fixed pricing strategy~\cite{faridani11}) requires estimation about the task acceptance probability mapping function $p(c)$ as input. However, these estimates may sometimes not be perfectly accurate. In this section, we examine the sensitivity of our pricing strategy to the estimation accuracy.


We train our dynamic pricing strategy under an inaccurate estimate of $p(c)$, and test it using the real value of $p(c)$. The task acceptance probability function is as Equation~\ref{eq_expt_pc}. For each experiment, we vary one parameter of the real $p(c)$ to examine the robustness of our dynamic pricing strategy. The estimation of other parameters are assumed to be accurate.

\begin{figure}[t]
\vspace{-5pt}
\centering
\subfigure[Remaining \# of tasks w.r.t. $s$]{\label{fig_remain_scale} \includegraphics[width = 1.6in]{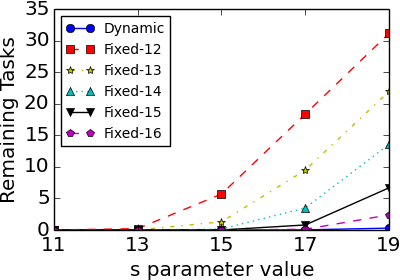}}
\subfigure[Average task reward w.r.t. $s$]{\label{fig_reward_scale} \includegraphics[width = 1.6in]{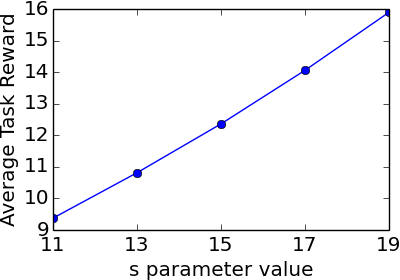}}
\vspace{-5pt}
\subfigure[Remaining \# of tasks w.r.t. $b$]{\label{fig_remain_bias} \includegraphics[width = 1.6in]{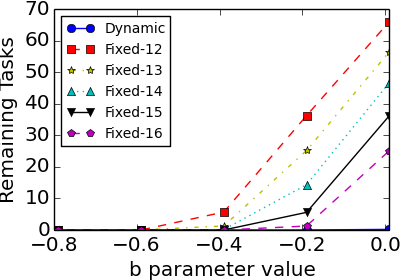}}
\subfigure[Average task reward w.r.t. $b$]{\label{fig_reward_bias} \includegraphics[width = 1.6in]{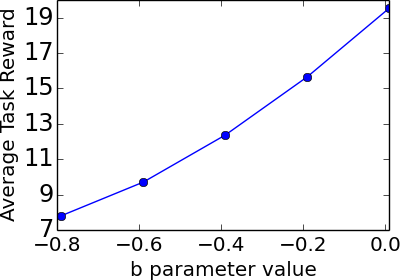}}
\subfigure[Remaining \# of tasks w.r.t. $M$]{\label{fig_remain_total} \includegraphics[width = 1.6in]{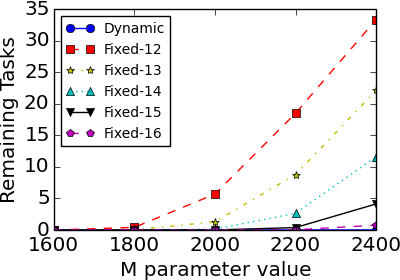}}
\subfigure[Average task reward w.r.t. $M$]{\label{fig_reward_total} \includegraphics[width = 1.6in]{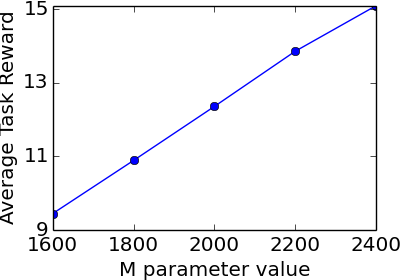}}
\vspace{-10pt}
\caption{Simulated average number of remaining tasks under inaccurate parameter estimation of $p(c)$ for dynamic pricing strategy and fixed pricing strategy(left) and average task reward for dynamic pricing strategy(right)}
\label{fig_sensitivity_parameter}
\vspace{-15pt}
\end{figure}

Figure~\ref{fig_sensitivity_parameter} shows the average number of remaining tasks (left figures) and average task reward (right figures) for our dynamic pricing strategy with respect to different values of parameters $s$, $b$, $M$ for the real $p(c)$. We focus first on the left figures, indicating the average number of remaining tasks. The data for the fixed pricing strategy (for various values of fixed price --- $12 \ldots 16$) is also added for comparison. Here, unlike the fixed price strategies that all have non-zero remaining tasks, the dynamic pricing strategy curve is not visible because the number of remaining tasks is very close to zero. {\em Thus, 
we see that our dynamic pricing strategy is much more robust under inaccurate parameter estimation compared to fixed pricing strategy: it returns 0 remaining tasks with very high probability, while the fixed pricing strategy completely fails to finish all the tasks on time.} The right figures (depicting only the dynamic pricing strategy) show how the dynamic pricing stays robust: as the parameters are increased, even though the dynamic pricing strategy has been learned on incorrect parameters, it automatically increases the task reward as necessary. 

%

%

\subsubsection{Sensitivity of Arrival-Rate Prediction}\label{sec_exp_arrival_rate}

Our dynamic pricing strategy and the fixed pricing strategy~\cite{faridani11} both require the prediction of future worker arrival rate. There will be some inevitable discrepancy between the predicted and actual arrival-rates because of the intrinsic variations in arrival-rate. In this section, we examine the stability of our dynamic pricing strategy against such discrepancies.

We divide the historical arrival-rate data retrieved from mturk-tracker into two separate parts: one for training and the other for testing. We train our pricing strategy using the training arrival-rate data, but apply it on the test arrival-rate data.
This way, we allow the algorithms to predict the general trend of the arrival-rate; however the algorithms will not be aware of the actual arrival-rate. We test our pricing strategies on $4$ different days in year 2014: 1/1, 1/8, 1/15, 1/22. The training arrival-rate is the average arrival-rate of the other $3$ days.
\begin{figure}[t]
\vspace{-5pt}
\centering
\subfigure[Average remaining \# of tasks for different testing days]{\label{fig_time_remain} \includegraphics[width = 1.6in]{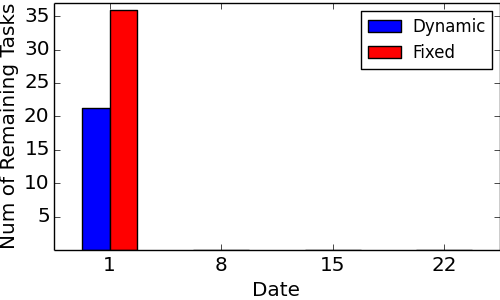}}
\subfigure[Average task reward for different testing days]{\label{fig_time_reward} \includegraphics[width = 1.6in]{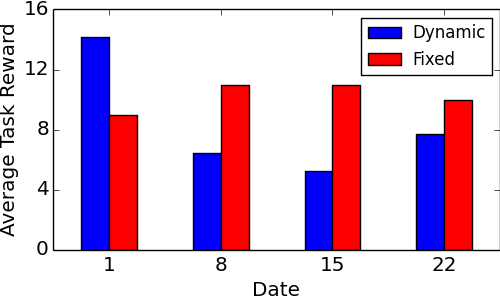}}
\subfigure[The actual arrival rate and training arrival rate for 1/1/2014]{\label{fig_arrival_101} \includegraphics[width = 1.6in]{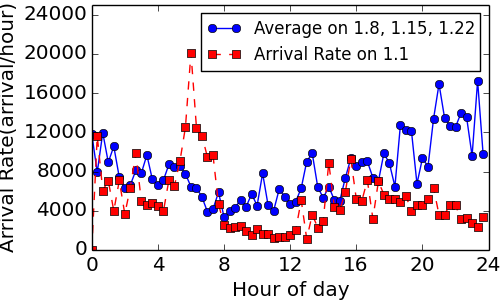}}
\subfigure[The actual arrival rate and training arrival rate for 1/22/2014]{\label{fig_arrival_122} \includegraphics[width = 1.6in]{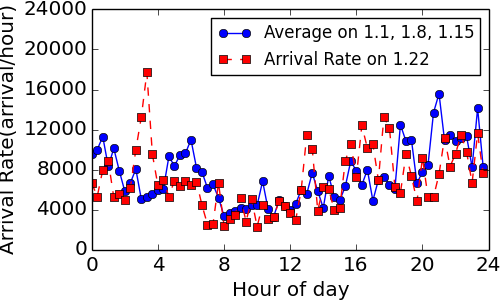}}
\vspace{-10pt}
\caption{Testing Sensitivity of Arrival Rates}
\label{fig_sensitivity_arrival_rate}
\vspace{-10pt}
\end{figure}

\eat{The confidence of completing all tasks is set to be $99.9$\% for both pricing strategies during training.} Figure~\ref{fig_time_remain} and~\ref{fig_time_reward} show the experiment results of the average number of remaining tasks and the average task reward respectively. As can be seen in the figures, both pricing strategies are relatively stable except for 1/1.

The surprising results for 1/1 can be explained by comparing Figure~\ref{fig_arrival_101} and~\ref{fig_arrival_122}. Figure~\ref{fig_arrival_122} shows the training arrival-rate and testing arrival-rate for 1/22: the training data is mostly in accordance with testing data, except that there are several random spikes in the testing data. The experiment results show that both pricing strategies are relatively stable to this kind of prediction error. On the other hand, Figure~\ref{fig_arrival_101} depicts a consistent deviation between training data and testing data on 1/1, in which case both pricing strategies performs poorly. Such a consistent deviation is probably due to the special date of 1/1, so these deviations shouldn't occur very frequently in practice. Naturally, the prediction of arrival-rate on special days is hard to do because it does not follow a normal weekday pattern. 
As a result, adaptive prediction techniques such as predicting the arrival-rate in next few hours based on arrival-rate in last few hours could be useful in such cases. We leave exploration of such adaptive schemes for future work.





\techreporttext{
\subsection{Fixed Budget Pricing Simulation}\label{sec_exp_static}

In this section, we simulate the static pricing strategy in Section~\ref{sec_fixed_budget} and study the distribution of finishing time. The experiment settings are as follows:
\begin{denselist}
\item
The total number of tasks $N = 200$, the total budget $B = 2500$ cents.
\item
The mapping function between task reward $c$ and task acceptance probability $p$ is still the same as Equation~\ref{eq_expt_pc}.
\item
Arrival-rates are retrieved from mturk-tracker as before.
\end{denselist}
Figure~\ref{fig_static_pricing} shows the simulation result. The average completion time is $23.2$ hours. However, any completion time between $18$ and $30$ hours is possible. Thus, the static pricing strategy does not try to guarantee any upper bound on the completion time but rather aims to minimize the completion time in expectation.

\begin{figure}[h]
\vspace{-5pt}
\centering
\includegraphics[width = 2.5in]{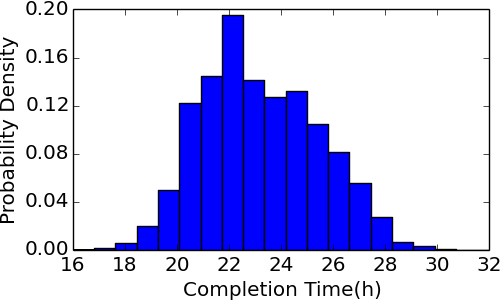}
\vspace{-5pt}
\caption{The simulated distribution of completion time}
\label{fig_static_pricing}
\vspace{-10pt}
\end{figure}

}

\subsection{Live Experiments on Mechanical Turk}\label{sec_exp_mturk}
In this section, we conduct experiments on Mechanical Turk to examine the effectiveness of our dynamic pricing strategy from Section~\ref{sec_fix_deadline} in practice. In Section~\ref{sec_exp_mturk_fixed} we first deploy the fixed pricing strategy on Mechanical Turk to collect data about worker arrival rate $\lambda(t)$ and task acceptance probability function $p(c)$. In Section~\ref{sec_exp_mturk_dynamic} we deploy our dynamic pricing strategy based on collected data, and experiment results are compared to the fixed pricing strategies. 
\techreporttext{In Section~\ref{exp_mturk_data_analysis}, we analyze the data collected from both experiments to provide some other interesting insights into workers' behaviors.}
Common experiment settings are listed below:
\begin{denselist}
\item 
We use an entity resolution task dataset from Joglekar et al.~\cite{confidence}.
Each task in the dataset consists of two photos (each with one athlete), and the worker is asked whether they contain the same person. 
In all experiments, we have 5,000 pairs of photos that we want workers to label.
\item
In all experiments, we post tasks on weekdays at 8 a.m. PST, with the deadline as 14 hours after start time (i.e., 10 p.m. PST). 
\item
The worker qualifications are: worker must have at least 90\% approval rate in history and live in United States.
\end{denselist}
Note that in Mechanical Turk, HITs (i.e., the unit of work in Mechanical Turk) with different price are grouped differently, even if they are issued by the same requester. Thus, workers looking for our specific HITs may not be able to know how many there are in total. So, in our experiments, we considered two options to vary price: (a) per HIT, keep the base price and number of tasks fixed, and vary the amount of bonus provided to the worker, or (b) per HIT, keep the base price fixed, and vary the number of tasks. We decided to go with the second option.
In our experiments, the price of each task group (i.e., HIT in Mechanical Turk) is fixed at \$0.02, and the price difference is expressed by the number of tasks (i.e., number of photo pairs to be labeled) in each HIT.

\subsubsection{Fixed Pricing Experiment}\label{sec_exp_mturk_fixed}
The fixed pricing experiment consists of five trials, where each HIT contains 10/20/30/40/50 tasks respectively. Given the total number of tasks in the trials is fixed at 5,000, the actual number of HITs posted to the marketplace is 500/250/167/125/100 respectively. In other words, in the five trials, the price for each task is implicitly \$0.002/0.001/0.00066/0.0005/0.0004 respectively. We stopped at 50 tasks per HIT, to limit worker fatigue.
Figure~\ref{fig_mturk_fixed_pricing} shows the number of HITs completed during the whole time period.

\begin{figure*}[t]
\centering
\vspace{-5pt}
\subfigure{\label{fig_mturk_fixed_pricing} \includegraphics[width = 2in]{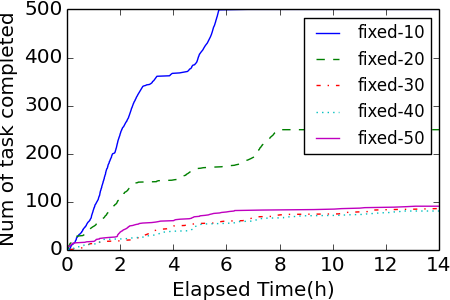}}
\subfigure{\label{fig_mturk_fixed_pricing_percentage} \includegraphics[width = 2in]{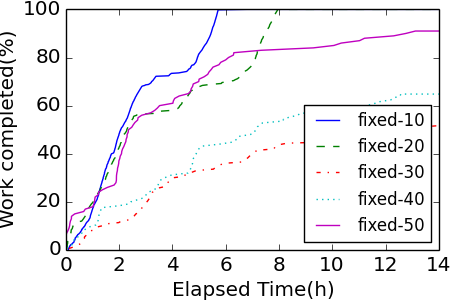}}
\subfigure{\label{fig_mturk_dynamic_pricing_percentage} \includegraphics[width = 2in]{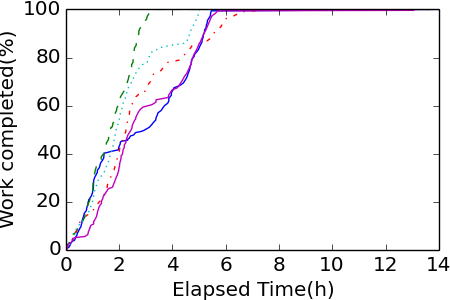}}
\vspace{-10pt}
\caption{Experiments on Mechanical Turk (a) The HIT completion rate for fixed pricing strategy (b) The percentage work completion rate for fixed pricing strategy (c) The percentage work completion rate for dynamic pricing strategy}
\label{fig_mturk}
\vspace{-15pt}
\end{figure*}


As can be seen from the figure, the HIT completion rate is positively correlated with the price of each task in general: for instance, when the elapsed time is 6 hours, the trial with 10 tasks per HIT has more than double the number of HITs completed than that with 20 tasks, and more than four times the number of HITs completed than that with 30, 40 or 50 tasks. When the number of tasks in each HIT is below 20, the task completion rate becomes high enough to have all tasks completed before the deadline (i.e., 14 hours). The task completion rates are very close for trials with grouping size 30/40/50, which can be explained by the small difference between unit task prices (\$0.00066 / \$0.0005 / \$0.0004).

However, the actual work completion rates (in terms of percentage of total work completed) are quite different when we take the \eat{question number} difference of number of tasks in each HIT into account, as shown in Figure~\ref{fig_mturk_fixed_pricing_percentage}.
Perhaps somewhat surprisingly, we see that after multiplying the number of tasks in each HIT to the quantities in Figure~\ref{fig_mturk_fixed_pricing}, the curve of the trial with grouping size $50$ becomes significantly higher than the curves of the trials with grouping size $30$ and $40$. This phenomenon suggest that grouping size per HIT has considerable effect on work completion rate: while workers choose HITs based on unit time wage, grouping more tasks into single HIT tends to force workers to stay on the same HIT for a longer time. (Note that in Mechanical Turk, workers do not earn any reward until they have completed all tasks in a HIT.)

\subsubsection{Dynamic Pricing Experiment}\label{sec_exp_mturk_dynamic}

To make experiment results comparable, the basic settings settings of the dynamic pricing experiment are the same as Section~\ref{sec_exp_mturk_fixed} (i.e., start time, deadline, total number of tasks), except that the grouping size is changed every hour based on our dynamic pricing strategy. The grouping size is chosen from 10/20/30/40/50, and the corresponding HIT acceptance rates are estimated from the fixed pricing experiment in the previous section. The worker arrival rates are estimated by averaging normalized worker arrival data in the five fixed pricing trials.


Figure~\ref{fig_mturk_dynamic_pricing_percentage} depicts work completion rate of the five trials (one on each day) in our experiment. As can be seen in the figure, the dynamic pricing strategy ends up completing all the tasks well before the deadline (6 hours instead of 14 hours). 
Furthermore, we find that the average total cost for the five trials is \$3.2, which is much less ($\approx 36\%$ less) than the total cost of \$5 for the fixed pricing strategy with grouping size $20$; in fact, note that the fixed pricing strategy with grouping size 20 had an elapsed time of 8, two hours more than the elapsed time of any of the trials for our strategies.

\papertext{In our extended technical report~\cite{treport}, we analyze the data collected via both experiments to shed light on two aspects of worker behavior: the accuracy of submitted answers and the number of tasks completed by each worker under different pricing settings. Our experimental results suggest that the quality of answers are reasonably good (with average accuracy $>90$\%) regardless of task price. While we refer the reader to the technical report for the details, Table~\ref{tbl_ave_acc_static_paper} and~\ref{tbl_ave_acc_dynamic_paper} depict the average accuracy for different group sizes for the fixed pricing experiment, and the average accuracy for different trials in the dynamic pricing experiment. Overall, there doesn't seem to be a distinct pattern governing how pricing affects accuracy, and we suspect the differences are all due to random variations.
Our results also show that the average number of tasks completed by each worker growing approximately exponentially with respect to task price. However, note that these results are preliminary, and that comprehensive experiments are necessary to draw further conclusions.}

\papertext{
	\begin{table}[h]
	\vspace{-5pt}
\centering
\scriptsize
\begin{tabular}{|c|c|c|c|c|c|}
\hline
\textbf{Group Size} & $10$ & $20$ & $30$ & $40$ & $50$\\
\hline \hline
\textbf{Average Accuracy} & $92.7$ & $90.4$ & $91.6$ & $90.0$ & $89.5$ \\
\hline 
\end{tabular}
	\vspace{-5pt}
\caption{The average accuracy of answers in the fixed pricing experiment}
	\vspace{-5pt}
	\label{tbl_ave_acc_static_paper}
\vspace{-10pt}
\end{table}
\begin{table}[h]
\centering
\scriptsize
\begin{tabular}{|c|c|c|c|c|c|}
\hline
\textbf{Trial} & $1$ & $2$ & $3$ & $4$ & $5$\\
\hline  \hline
\textbf{Overall Ave. Accuracy} & $90.7$ & $91.7$ & $88.2$ & $95.0$ & $90.9$ \\
\hline
\end{tabular}
	\vspace{-5pt}
\caption{The average accuracy of answers in dynamic pricing experiment}
	\vspace{-5pt}
\label{tbl_ave_acc_dynamic_paper}
\vspace{-10pt}
\end{table}	
}

\techreporttext{
\subsubsection{Analysis of Collected Data}\label{exp_mturk_data_analysis}
In this section, we further analyze the data collected from previous experiments to study the behavior of workers.
\begin{figure}[h]
\centering
\includegraphics[width = 2.5in]{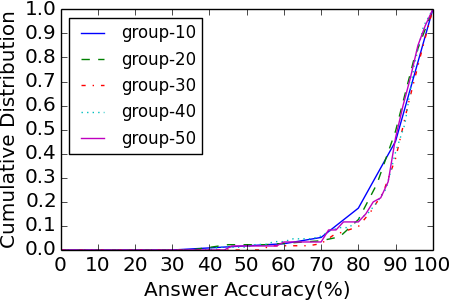}
\vspace{-5pt}
\caption{Answer quality under different prices in fixed pricing experiment}
\label{fig_mturk_static_quality}
\vspace{-10pt}
\end{figure}
\begin{figure}[h]
\centering
\includegraphics[width = 2.5in]{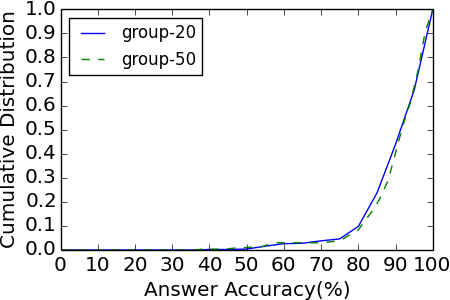}
\vspace{-5pt}
\caption{Answer quality under different prices in dynamic pricing experiment}
\label{fig_mturk_quality}
\vspace{-10pt}
\end{figure}

\begin{table}[h]
\centering
\scriptsize
\begin{tabular}{|c|c|c|c|c|c|}
\hline
\textbf{Group Size} & $10$ & $20$ & $30$ & $40$ & $50$\\
\hline \hline
\textbf{Average Accuracy} & $92.7$ & $90.4$ & $91.6$ & $90.0$ & $89.5$ \\
\hline 
\end{tabular}
\caption{The average accuracy of answers in the fixed pricing experiment}
\label{tbl_ave_acc_static}
\vspace{-10pt}
\end{table}
\begin{table}[h]
\centering
\scriptsize
\begin{tabular}{|c|c|c|c|c|c|}
\hline
\textbf{Trial} & $1$ & $2$ & $3$ & $4$ & $5$\\
\hline  \hline
\textbf{Ave. Accuracy w/ group size $20$} & $92.9$ & $89.1$ & $89.4$ & $94.9$ & $92.0$ \\
\hline 
\textbf{Ave. Accuracy w/ group size $50$} & $89.8$ & $92.8$ & $87.4$ & $95.2$ & $90.2$ \\
\hline
\textbf{Overall Ave. Accuracy} & $90.7$ & $91.7$ & $88.2$ & $95.0$ & $90.9$ \\
\hline
\end{tabular}
\caption{The average accuracy of answers for tasks with group size $20$ and $50$ in dynamic pricing experiment}
\label{tbl_ave_acc_dynamic}
\vspace{-10pt}
\end{table}


Figure~\ref{fig_mturk_static_quality} and Figure~\ref{fig_mturk_quality} depict the cumulative distribution of accuracy of worker's answers under different task price settings in fixed pricing and dynamic pricing experiments respectively. The two curves in the dynamic pricing experiment are the cumulative accuracy distribution of answers when the dynamic pricing strategy picks a grouping size of $20$ or $50$. We only plot these two curves for the dynamic pricing case, because the other grouping sizes are rarely used by the dynamic pricing strategy in our experiments.
Overall, we find that the curves (the five in Figure~\ref{fig_mturk_static_quality} and the two in Figure~\ref{fig_mturk_quality}) are all very similar, indicating that the pricing does not affect quality much.
Note that the group-$50$ curve (and the group-$40$ curve too, to some extent) in the fixed pricing plot appears ``jagged'', while the other group sizes have a smoother plot. This is probably because there are fewer tasks in that trial (recall that we are fixing the total number of questions, so when the grouping size increases, the total number of tasks decreases), and the number of possible accuracy values is larger (which means that curve will be less smooth when connecting points to draw the cumulative distribution curve). The average accuracy of answers in the two experiments are shown in Table~\ref{tbl_ave_acc_static} and Table~\ref{tbl_ave_acc_dynamic} respectively.

The experimental results show that the average accuracy of answers are all reasonably good (higher than or close to 90\% accuracy), and their differences are not statistically significant. This result suggests that, under our experimental conditions, pricing mainly affects whether workers choose to work on the HIT or not. If they decide to work on one of our HITs, the answers then provided are reasonably good. Studying the general correlation between task price and answer quality requires additional in-depth experiments, which are beyond the scope of this paper.
\begin{figure}[h]
\centering
\includegraphics[width = 2.5in]{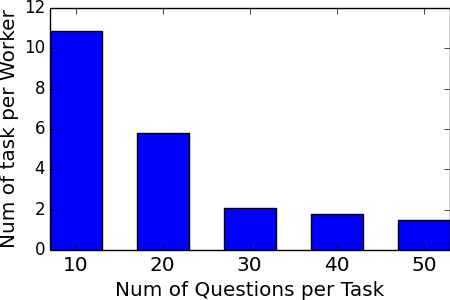}
\vspace{-5pt}
\caption{Average number of tasks completed by each worker}
\label{fig_mturk_work_length}
\vspace{-10pt}
\end{figure}
 
Figure~\ref{fig_mturk_work_length} shows the average number of HITs completed by each worker under different pricing settings in the fixed pricing experiments in Section~\ref{sec_exp_mturk_fixed}. As shown in the figure, with the lower task price, workers tend to leave after they completed one or two HITs. On the other hand, when task price is higher, some workers will tend to continuously work on the same kind of task. Note that the NHPP model in Faridani's work~\cite{faridani11} does not explicitly model this phenomenon. By incorporating this behavior into NHPP model, the worker-arrival rate could be predicted more accurately, and this could potentially improve the effectiveness of our dynamic pricing strategy.
}

\techreporttext{
\section{Discussion}\label{sec:discussion}

In this section, we discuss some possible generalizations of our pricing schemes, as well as some possible impact of our dynamic pricing strategy on worker behavior.

\subsubsection*{Multiple Task Types} In some cases, we may have multiple types of tasks that need to be completed by a certain deadline. For instance, we may have 100 categorization tasks, and 500 labeling tasks that all need to be completed at the same time. Multiple task types is easy to incorporate into our model. We simply represent the state as a vector $(n_1, n_2, \ldots, n_k, t)$,
where $n_i$ represents the number of tasks of type $i$. The resulting objectives, relationships, and dynamic programming-based optimization algorithms are similar.

\subsubsection*{Incorporating Quality Control for Filtering Tasks} 

In \cite{crowdscreen}, we describe an MDP based technique to optimize for cost and accuracy, specifically for filtering or rating tasks. At a high level, our algorithm from \cite{crowdscreen}, given an accuracy threshold overall, generates a per-item quality-control strategy guaranteeing specified accuracy with the minimum total number of questions in expectation (ignoring pricing per question). A quality-control strategy is represented using a collection of points $(x, y)$, where $x$ is the number of No answers for that item and $y$ is the number of Yes answers, and decisions, and for each point there is a decision associated with that point, either continue asking questions for that task, or stop and return PASS/FAIL --- representing the fact that the task either satisfied the filtering predicate or did not. Now, we can generalize this approach as well as the approach described in this paper to give a solution optimized for cost, latency, and accuracy.

Consider the problem where we once again have $N$ filtering tasks that we need to complete, and our goal is to minimize expected cost, while ensuring that accuracy is within threshold, and that tasks are completed by a certain deadline. As a first step, we generate the per-task quality-control strategy using algorithms from \cite{crowdscreen}, ensuring the minimum number of questions are used, while guaranteeing that the accuracy is within threshold. Let this quality-control strategy have $k$ points $(x, y)$, corresponding to $k$ different combinations of the number of \# of Y and N answers for that task---we then use the current technique described in this paper for the state space represented by:
\begin{quote}
(\# of undecided tasks at pt.~1, \# of undecided tasks at pt.~2, $\ldots$, \# of undecided tasks at pt.~$k$, remaining time) 
\end{quote}
instead of
\begin{quote}
(\# of tasks, remaining time)
\end{quote}
and then optimize for the best pricing technique.
Note that whenever an task is completed (i.e., the quality-control strategy deems that we don't need any additional answers for that task), it is removed from the set of undecided tasks, and no longer counts towards the pricing state space described above. 

The DP algorithm to determine the optimal strategy overall is similar to the one described in Section~\ref{sec_fix_deadline}, except that it is computed over this new state space.
We let $P$ denote the vector $(n_1, \ldots, n_k)$, representing the number of tasks in each of the $k$ points of the quality-control strategy, and $P'$ denote the vector corresponding to the number of tasks in each of the $k$ points (in the quality-control strategy) that the tasks transitioned to by the next time interval. we let $s$ represent the number of additional answers between the two points (that is, the number of additional Yes/No answers).
Then, the probabilities are as follows:
\begin{align*}
\mathbf{Pr}\{(P, t) \rightarrow (P', t + 1)|c_{P,t}\} & = \\ 
\mathbf{Pr}(P \rightarrow P' | s, P) & \cdot \mathbf{Pois}(s|\lambda = \lambda_t p(c_{P,t}))
\end{align*}
The latter term is as Equation~\ref{eq:probtrans}, while the first term can be computed using probability machinery from \cite{crowdscreen}. Overall, the complexity is: $O(N^{2k} N_T C)$, which can be large if $k$ is large. (Note that when $k = 1$, we default to the setup from Section~\ref{sec_fix_deadline}). 
Thus, our problem is fundamentally challenging if $k$ is large.
Typically, though, $k$ is often as small as $9$, say when a small majority vote quality-control strategy is used.

Recognizing the fact that the problem may be intractable when $k$ is large, we present next some approximate techniques for this problem. It remains to be seen which of these techniques would be better suited for the problem and lead to better approximations. While the first technique has guarantees (but only asymptotically) the second technique does not have any guarantees, but is more tractable and easier to understand:
\squishlist
\item {\bf Representing Using Posterior Probabilities:} For the cases when $k$ is large, we can approximate the quality-control strategy generation process, as described in \cite{rate}, where we map points in $(x, y)$ to the segments in the real line: 
$$[0, a) [a, 2a) \ldots [1-a, 1]$$
where the interval $[a, 2a)$ represents the fact that the posterior probability of the item satisfying the filter is between $a$ and $2a$.
Thus, given a point $(x, y)$ will map to a segment $[ia, (i+1)a)$ along this real line, where 
$$ia \leq \textbf{Pr}[\text{item is a \ } 1| (x, y)] < (i + 1) a$$
We then regard every point $(x, y)$ that maps to such an interval $[ia, (i + 1) a)$ as having posterior probability $ia + a/2$.
The algorithm stays unchanged, except that the $k$ points are now represented approximately by these $1/a$ intervals;
and the complexity is now $O(N^{2/a}N_TC)$.
As shown in \cite{rate}, as $a \rightarrow 0$, the optimal strategy in this representation (with intervals) tends to the optimal strategy in the old representation (with points) asymptotically.
This argument follows from standard arguments for discretizing continuous state markov decision processes, e.g., 
\cite{bertsekas1995dynamic}.
\item {\bf Keeping Quality Control Separate from Pricing Optimization:} The second approximation technique  treats quality control as an orthogonal problem. Once we compute the quality-control strategy, for each point $(x, y)$ in the quality-control strategy from \cite{crowdscreen}, we can compute the worst case additional number of questions. For instance, if the strategy extends all the way to $x + y = 5$, then the worst case additional number of questions at $(0, 0)$ is $5$. On the other hand, the worst case additional number of questions at $(2, 1)$ may be just 1 if both $(3, 1)$ and $(2, 2)$----which are the two points reachable from $(2, 1)$ on getting an additional answer---are end states where a decision of PASS/FAIL is made. 
Then, we can apply our technique from Section~\ref{sec_fix_deadline} to the problem with $N' = N \times \alpha$, where
$\alpha$ is the worst case additional number of questions from the origin.
Now, we have a strategy designed from $(0, 0)$ to $(N' = N \alpha, T)$ where $N'$ represents the total worst case number of questions across all tasks.
We then run the strategy as before, while implicitly moving each task on the quality control strategy as well, and having that influence the $N'$ (i.e., the total worst case number of questions across all tasks) corresponding to where the strategy is currently at.
That is:
$$N' = \sum_{\text{all tasks \ }i}{\text{worst case additional questions at \ } P(i)}$$
where $P(i)$ denotes the point on the quality control strategy that the task is at.
We explain this using an example. Let there be 10 tasks, and let the quality-control strategy we desire to use be a majority vote strategy with 3 questions (i.e., ask 3 questions and take the majority).
Then, the worst case number of questions at point $(0, 0)$ in the quality-control strategy will be $3$.
So we will begin the strategy at $(10 \times 3 = 30, 0)$. 
After some time, let 5 tasks be at $(1, 1)$, while 2 reach $(2, 0)$ and 3 reach $(0, 2)$.
In the first case, the worst case additional number of questions is 1, while the other two tasks have worst case additional number of questions as 0.
Thus, overall, we are now at $(5 \times 1 + 2 \times 0 + 2 \times 0 = 5, T_i)$, where $T_i$ is the current time.
The reason we use the worst case additional number of questions is to be conservative in order to meet the deadline on time, at potentially additional cost. We could instead use the expected additional number of questions, but we may end up not meeting the deadline.
The complexity of the DP algorithm for this technique is very reasonable: $O((Nk)^2 N_T T)$. (Note that the worst case number of questions from the origin can at most be $k$, but could be much smaller.)
\squishend

If we had a prior distribution on difficulty, we can easily take that into account in the quality-control strategy, as described in \cite{rate}.

\subsubsection*{Optimizing Tradeoff between Deadline and Budget} 


In some scenarios, we may have neither a fixed deadline or a fixed budget, and we may want to achieve some optimal tradeoff between the two. We now focus on optimizing a linear combination of expected cost and time. 
Thus, our objective is now:
$$Q = \mathbb{E}(\textbf{cost}) + \alpha\mathbb{E}(\textbf{latency})$$
We consider two scenarios: the first, which makes more assumptions, and a more general scenario next. The first scenario will act as a ``building block'' for the second.

\stitle{Fixed Rate:}
We first focus on optimizing the objective under the assumption that the rate at which workers appear in the marketplace is fixed at $\lambda$, i.e., $\lambda(t) = \lambda, \forall t$. 
(This assumption is a bit more drastic than the assumption in Section~\ref{sec_fixed_budget}, where we assumed that the rate is not fixed but is constant over a long period.)
Given that we are discretizing time units as multiples of 1, the rate at which workers appear in the marketplace is the same as the expected number of workers who appear in the marketplace in a unit time interval. 

\begin{figure}[ht!]
\centering
\includegraphics[width = 2.5in]{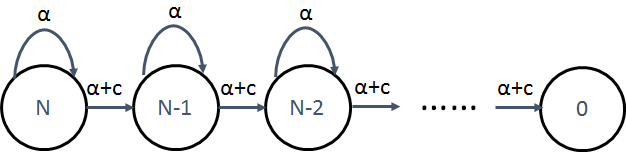}
\vspace{-5pt}
\caption{State transition diagram}
\label{fig_linear_state_transition}
\vspace{-10pt}
\end{figure}
Turns out, under such a scenario, we do not need to record cost or time, since we do not have a deadline, and since the amount of time elapsed or cost consumed until one gets to a given state is ``sunk cost/time''.
Thus, the states are simply recorded using $(n)$, where $n$ is the number of outstanding tasks. We let $c_n$ denote our price for all tasks at time $n$. 
Our time interval discretization will be set to be small enough that the likelihood of two tasks being performed within a time interval is nearly zero.
We depict the set of states and transitions in Figure~\ref{fig_linear_state_transition}. From any state, we either stay in that state (if no tasks are picked up), or move to the neighboring state on the right (if one task is picked up --- recall that since our time granularity is set to be as small so that we do not ever end up having more than one task picked). The costs for transitions are labeled on the edges, and are described more in the equations below:
\begin{align*}
\mathbf{Pr}\{(n) \rightarrow (n - 1)|c_{n}\} & =  e^{-\lambda p(c_{n})} \lambda p(c_{n}) \\
\mathbf{cost}\{(n) \rightarrow (n - 1)|c_{n}\} & =  c_n + \alpha \\ 
\mathbf{cost}\{(n) \rightarrow (n)|c_n\} & = \alpha  
\end{align*}
These equations are simply counterparts to the equations~\ref{eq:probtrans} and~\ref{eq:costtrans} in Section~\ref{sec_fix_deadline}, with the additional restriction that the probability of transitioning from $n$ to $n-s$ where $s > 1$ is 0, and the fact that the cost of transitioning back to the same state is $\alpha$, i.e., $\alpha \times \text{latency}$, which is 1, and the cost of transitioning to a neighboring state is $c_n + \alpha \times \text{latency} = c_n + \alpha$. 
Now, following the dynamic programming recipe from Section~\ref{sec_fix_deadline}, 
\begin{align*}
\mathbf{Opt}(n) = \min_c  [\mathbf{Opt}(n - 1) + c + \alpha] & \times  \mathbf{Pr}\{(n) \rightarrow (n - 1)|c\}  \\
						+  \ \ \ [ \mathbf{Opt}(n) + \alpha] & \times   \mathbf{Pr}\{(n) \rightarrow (n)|c\} 
\end{align*}
\noindent Here, given the price $c$ that is set at state $n$, either one task is completed by the next time interval, or no task is completed, which accounts for the two quantities within the minimization. In the latter case, we stay at the same state. Note here that we need to, for each $c$, solve for $\mathbf{Opt}(n)$, and then pick the smallest value. That is the best price for state $(n)$. The complexity of this procedure is dependent, once again, on the number of price choices $C$: the complexity is simply: $O(N C)$.

\stitle{Relaxing the Linearity Assumption:}
The above techniques can be generalized to the more realistic assumption that we made in Section~\ref{sec_fixed_budget}, that the expected total latency $T$ is linearly correlated with worker arrival quantity $W$ (See Section~\ref{sec_linear_assumption} for the justification of this):
$$ \mathbb{E}[T|W] = \frac{W}{\bar \lambda} $$
where $\bar \lambda$ is the average arrival rate of workers. Now it follows that:
$$ \mathbb{E}[T] = \frac{1}{\bar \lambda} \mathbb{E}[W] $$
Substituting the above equation into our objective function, we get:
\begin{align*}
Q & = \mathbb{E}(\textbf{cost}) + \alpha\mathbb{E}(\textbf{latency}) \\ & = \mathbb{E}(\textbf{cost}) + \alpha \frac{1}{\bar \lambda} \mathbb{E}(\textbf{worker arrival}) 
\end{align*}
Optimizing this new objective function is very similar to optimizing the original objective function: the state space is still the same, except that transitions between states are slightly different, here we have:
\begin{align*}
\mathbf{Pr}\{(n) \rightarrow (n - 1)|c_{n}\} & =  p(c_{n}) \\
\mathbf{cost}\{(n) \rightarrow (n - 1)|c_{n}\} & =  c_n + \alpha \frac{1}{\bar \lambda} \\
\mathbf{cost}\{(n) \rightarrow (n)|c_n\} & = \alpha \frac{1}{\bar \lambda} 
\end{align*} 
where each transition represents one single worker arrival event, unlike the previous scenario, where transitions happened at each time interval. Here, once a worker arrives, whether or not they choose to work on our task determines if we transition to the neighboring state, or we stay at the same state. We can see that the state transition diagram is the same as before, except a few coefficients are different now. Therefore, the same dynamic programming technique can also be applied here to find the optimal solution, and the complexity stays the same, i.e., $O(NC)$.

\subsubsection*{Long-term Impact on Worker Behavior} As in any marketplace / game theoretic scenario, some workers may learn to ``game'' the system as well as our dynamic pricing algorithm. This is inevitable. In practice, however, we expect that as long as the majority of workers are part-time workers (which is certainly true in current marketplaces), they are not likely to witness our algorithm in action and evolve their decisions to take advantage of it. Furthermore, even if many workers know about our dynamic pricing algorithm and wish to take advantage of the system, as the price increases, some workers may decide to work on all the remaining tasks at that price, leaving other workers to not have any work (and hence rethink their strategy). As long as the workers in the marketplace are not cooperating with each other, we expect the system to achieve certain equilibrium in the end. Lastly, we expect our pricing techniques to be updated once in a while to reflect the current state of the marketplace---for instance, if workers no longer pick up \$0.1 tasks, we may want to offer our minimum price as \$0.2.
}


\section{Related Work}\label{sec:related}

The prior work related to ours can be placed in a few categories; we describe each of them in turn:

\smallskip
\noindent {\bf Pricing Schemes:}
Faridani~\cite{faridani11} develop models for marketplace dynamics that we leverage in this paper. They also develop static pricing strategies that we compare against. To the best of our knowledge, there has been no work on optimizing price apart from \cite{faridani11}.

\smallskip
\noindent {\bf Control Theory:}  Recent work has leveraged decision theory for improving cost and quality in simple crowdsourcing workflows, typically using POMDPs (Partially Observable MDPs): Dan Weld's group has designed strategies to dynamically choose the best decision to make at any step in the workflow (refine, improve, vote, or stop), and also to dynamically switch between workflows to improve the overall ``utility''~\cite{clowder,DBLP:conf/uai/LinMW12,DBLP:conf/aaai/LinMW12,workflow}. Kamar et al.~\cite{DBLP:conf/aamas/KamarHH12} use POMDPs to study how to best utilize participation in voluntary crowdsourcing systems, specifically, Galaxy Zoo, an astronomical data set verified by human workers. The papers mentioned above do not provide theoretical guarantees. Our prior work also uses decision theory for getting guarantees on cost and accuracy for filtering~\cite{rate, crowdscreen}. None of these prior papers study the problem of determining optimal pricing for tasks over time: all of them assume that each task has a fixed price or reward, and optimize the set of tasks to meet accuracy guarantees.

\smallskip
\noindent {\bf Crowd Algorithms:} There has been a lot of recent activity centered around designing data processing algorithms where the unit operations are performed by human workers, such as filtering~\cite{crowdscreen}, sorting and joins~\cite{so-who-won, markus-sorts-joins}, top-$k$~\cite{top-k}, deduplication and clustering~\cite{DBLP:journals/pvldb/WangKFF12, crowdclustering}
 and categorization~\cite{humangs}. None 
of these papers explore the problem of pricing tasks to complete on time.

Of these papers, just categorization~\cite{humangs} and filtering~\cite{crowdscreen, rate} consider the aspect of latency, and there too, they use number of round-trips as a proxy for latency rather than the true elapsed time.


\smallskip
\noindent {\bf Error Estimation:} There has been significant work on simultaneous estimation of answers to tasks and errors of workers using the EM algorithm or other local optimization techniques. There have been a number of papers studying increasingly expressive models for this problem, including difficulty of tasks and worker expertise~\cite{whitehill-accuracy, raykar-whom-to-trust, LPI}, adversarial behavior~\cite{DBLP:journals/jmlr/RaykarY12}, and online evaluation of workers~\cite{online-crowdsourcing, swiftly, cdas}. 
There has also been work on choosing workers for evaluating different items so as to reduce overall error rate~\cite{get-another-label, donmez-learning-accuracy}. 
Recent work has also tried to obtain theoretical guarantees for both worker error estimates as well as correct labels for items~\cite{DDKR, GKM, KOS, TR}. 
Our work on pricing tasks is orthogonal to this line of work, and can be combined with any of these schemes to better price
a batch of tasks to complete by a given deadline.

\smallskip
\noindent {\bf Applications:} There are a number of useful applications of crowdsourcing, such as sentiment analysis~\cite{rion-snow}, identifying spam~\cite{savage-spam}, determining search relevance~\cite{alonso-relevance-judgements}, and translation~\cite{zaidan-translation}. 

\section{Conclusions}\label{sec_conclusion}

In this paper, we developed algorithms to optimally set and vary the price for human computation tasks in a crowdsourcing marketplace to meet latency and cost constraints. 
For a monetary budget scenario, we demonstrated that static pricing strategies lead to optimal completion times, and developed efficient algorithms to find near-optimal pricing strategies. For a fixed deadline scenario, we demonstrated that our techniques, based on MDPs, outperform fixed pricing strategies by up to 30\% on simulations based on real-world crowdsourcing marketplace data and live experiments, and are more robust to errors in estimates of marketplace parameters and predictions of future trends. Our techniques can be profitably employed in scenarios demanding the repeated use of crowdsourcing on a large scale, wherein the cost savings will be massive. 


%
{\scriptsize
\bibliographystyle{abbrv}
\bibliography{sigproc,crowdbib}  
}
%
%
\end{document}